%% file: verification.tex
\documentclass[11pt,envcountsect]{llncs}
\usepackage[margin=1in]{geometry}
\usepackage{amsmath,amssymb,xcolor,times}
\usepackage[final]{graphicx}
\usepackage{algorithm}
\usepackage[noend]{algpseudocode}
\floatname{algorithm}{Mechanism}

\def\eps{\varepsilon}
\def\union{\cup}

\def\cut{\cap}

\def\floor#1{\mathop{\left\lfloor#1\right\rfloor}}

\def\Prob{\mathbb{P}\mathrm{r}}
\def\Exp{\mathbb{E}}

\newcommand{\reals}{\mathbb{R}}

\def\abs#1{\left|#1\right|}
\def\OPT{\mbox{\sc opt}}
\def\Ver{\mathrm{Ver}}
\def\ExpVer{\Exp\mathrm{Ver}}
\def\cost{\mathrm{cost}}
\def\e{\epsilon}
\def\norm#1{\|#1\|}
\def\Pow{\mathrm{Pow}}
\def\PPow{\mathrm{PartPow}}
\def\Expo{\mathrm{Expo}}

\def\NP{\mathbf{NP}}

\newcommand{\matr}[1]{\vec{#1}}
\newcommand{\ver}{\mathrm{ver}}
\def\exp{\mathrm{exp}}
\spnewtheorem*{proofsketch}{Proof sketch}{\itshape}{\rmfamily}

\title{Who to Trust for Truthfully Maximizing Welfare?}
\author{Dimitris Fotakis\inst{1} \and Christos Tzamos\inst{2}
\and Emmanouil Zampetakis\inst{2}}

\institute{%
National Technical University of Athens, 157 80 Athens, Greece.\\
\and
Massachusetts Institute of Technology, Cambridge, MA 02139\\
\email{fotakis@cs.ntua.gr}, \email{tzamos@mit.edu}, \email{mzampet@mit.edu}}

\begin{document}

\maketitle

\begin{abstract}
We introduce a general approach based on \emph{selective verification} and obtain approximate mechanisms without money for maximizing the social welfare in the general domain of utilitarian voting. Having a good allocation in mind, a mechanism with verification selects few critical agents and detects, using a verification oracle, whether they have reported truthfully. If yes, the mechanism produces the desired allocation. Otherwise, the mechanism ignores any misreports and proceeds with the remaining agents. We obtain randomized truthful (or almost truthful) mechanisms without money that verify only $O(\ln m/\eps)$ agents, where $m$ is the number of outcomes, independently of the total number of agents, and are $(1-\eps)$-approximate for the social welfare. We also show that any truthful mechanism with a constant approximation ratio needs to verify $\Omega(\log m)$ agents. A remarkable property of our mechanisms is \emph{robustness}, namely that their outcome depends only on the reports of the truthful agents.
\end{abstract}

\thispagestyle{empty}%
\setcounter{footnote}{0}%
%\newpage
\pagestyle{plain}
\pagenumbering{arabic}
\setcounter{page}{0}%

%\category{F.2.2}{Theory of Computation}{Analysis of Algorithms and Problem Complexity}
%\category{J.4}{Computer Applications}{Social and Behavioral Sciences}[Economics]
%\terms{Algorithms; Theory; Economics}
%\keywords{Algorithmic Mechanism Design; Mechanisms with Verification; Approximate Mechanism Design without Money}

\input{intro}
\input{model}

\input{facility}

\input{power}

\input{lowerbound}
\input{impossibility}

\input{partial}

\input{exponential}

% Bibliography
\newpage
\bibliographystyle{plain}
\bibliography{mechdesign}

\input{appendix}
\input{appendix_facility}
\input{appendix_power}

\input{appendix_impossibility}
\input{appendix_partial}

\input{appendix_exponential}

\end{document}

%% file: intro.tex
\section{Introduction}
\label{s:intro}

Let us start with a simple mechanism design setting where we place a facility on the line based on the preferred locations of $n$ strategic agents. Each agent aims to minimize the distance of her preferred location to the facility and may misreport her location, if it finds it profitable. Our objective is to minimize the maximum distance of any agent to the facility and we insist that the facility allocation should be \emph{truthful}, i.e., no agent can improve her distance by misreporting her location.
The optimal solution is to place the facility at the average of the two extreme locations. However, if we cannot incentivize truthfulness through monetary transfers (e.g., due to ethical or practical reasons, see also \cite{PT09}), the optimal solution is not truthful. E.g., the leftmost agent has an incentive to declare a location further on the left so that the facility moves closer to her preferred location. In fact, for the infinite real line, the optimal solution leads to no equilibrium declarations for the leftmost and the rightmost agent.
The fact that in this simple setting, the optimal solution is not truthful was part of the motivation for the research agenda of \emph{approximate mechanism design without money}, introduced by Procaccia and Tennenholtz \cite{PT09}. They proved that the best deterministic (resp. randomized) truthful mechanism achieves an approximation ratio of $2$ (resp. $3/2$) for this problem.

Our work is motivated by the simple observation that the optimal facility allocation can be implemented truthfully if we inspect the declared locations of the two extreme agents and verify that they coincide with their preferred locations (e.g., for their home address, we may mail something there, visit them or ask for a certificate). Inspection of the two extreme locations takes place before we place the facility. If both agents are truthful, we place the facility at their average. Otherwise, we ignore any false declarations and recurse on the remaining agents. This simple modification of the optimal solution is truthful, because non-extreme agents do not affect the facility allocation, while the two extreme agents cannot change the facility location in their favor, due to the verification step.
Interestingly, the Greedy algorithm for $k$-Facility Location (see e.g., \cite[Sec.~2.2]{WS10}) also becomes truthful if we verify the $k$ agents allocated a facility and ignore any liars among them (see Section~\ref{s:facility}). Greedy is $2$-approximate for minimizing the maximum agent-facility distance, in any metric space, while \cite{FT12} shows that there are no deterministic truthful mechanisms (without verification) that place $k \geq 2$ facilities in tree metrics and achieve a bounded (in terms of $n$ and $k$) approximation ratio.

\smallskip\noindent{\bf Selective Verification: Motivation and Justification.}
Verifying the declarations of most (or all) agents and imposing large penalties on liars should suffice for the truthful implementation of socially efficient solutions (see e.g., \cite{CESY12}). But in the facility location examples above, we truthfully implement the optimal (or an almost optimal) solution by verifying a very small number of agents (independent of $n$) and by using a mild and reasonable penalty. Apparently, verification is successful in these examples because it is \emph{selective}, in the sense that we verify only the critical agents for the facility allocation and fully trust the remaining agents.

Motivated by this observation, we investigate the power of \emph{selective verification} in approximate mechanism design without money in general domains. We consider the general setting of \emph{utilitarian voting} with $m$ outcomes and $n$ strategic agents, where each agent has a nonnegative utility for each outcome. We aim at truthful mechanisms that verify few critical agents and approximate the maximum \emph{social welfare}, i.e., the total utility of the agents for the selected outcome. Our goal is to determine the best approximation guarantee achievable by such mechanisms with limited selective verification, so that we obtain a better understanding of the power of limited verification in mechanism design without money. Our main result is a smooth and essentially best possible tradeoff between the approximation ratio and the number of agents verified by randomized truthful (or almost truthful) mechanisms with selective verification.

Our general approach is to start from a (non-truthful) allocation rule $f$ with a good approximation guarantee for the social welfare and to devise a mechanism $F$ without money that incentivizes truthful reporting by selective verification. The mechanism $F$ first selects an outcome $o$ and an (ideally small) verification set of agents according to $f$ (e.g., for facility location on the line, the allocation rule $f$ is to take the average of the two extreme locations, the selected outcome $o$ is the average for the particular instance and the verification set consists of the two extreme agents). Next, $F$ detects, through the use of a \emph{verification oracle}, whether the selected agents are truthful. If yes, the mechanism outputs $o$. Otherwise, $F$ excludes any misreporting agents and continues with the remaining agents.
We note that $F$ asks the verification oracle for a single bit of information about each agent verified: whether she has reported truthfully or not. $F$ excludes misreporting agents from the allocation rule, so it does not need to know anything else about their true utilities.

Instead of imposing some explicit (i.e., monetary) penalty to the agents caught lying by verification, the mechanism $F$ just ignores their reports, a reasonable reaction to their revealed attempt of manipulating the mechanism. We underline that liars still get utility from the selected outcome. It just happens that their preferences are not taken into account in the allocation. For these reasons, the penalty of exclusion from the mechanism is mild and compatible with the spirit of mechanisms without money.

Selective verification allows for an explicit quantification of the amount of verification and is applicable to essentially any domain. From a theoretical viewpoint, we believe that it can lead to a deep and delicate understanding of the power of limited verification in approximate mechanism design without money. From a practical viewpoint, the extent to which selective verification and the penalty of ignoring false declarations are natural very much depends on the particular domain / application. E.g., for applications of facility location, where utility is usually determined by the home address of each agent, public authorities have simple ways of verifying it. E.g., registration to a public service usually requires a certificate of address. Failure to provide such a certificate usually implies that the application is ignored, with no penalties attached.

\smallskip\noindent{\bf Technical Approach and Results.}
A (randomized) mechanism with selective verification is \emph{truthful} (in expectation) if no matter the reports of the other agents and whether they are truthful or not, truthful reporting maximizes the (expected) utility of each agent from the mechanism. Two nice features of our allocation rules (and mechanisms) is that they are \emph{strongly anonymous} and \emph{scale invariant}. The former means that the allocation only depends on the total agents' utility for each outcome (and not on each agent's contribution) and the latter means that multiplying all valuations by a positive factor does not change the allocation.

For mechanisms with selective verification, truthfulness is an immediate consequence of two natural (and desirable) properties: robustness and voluntary participation.
Robustness is a strong property made possible by selective verification. A mechanism with verification $F$ is \emph{robust} if $F$ completely ignores any misreports and the resulting probability distribution is determined by the reports of truthful agents only. So, if $F$ is robust, no misreporting agent can change the resulting allocation whatsoever.
We achieve robustness through obliviousness of $F$ to the declarations of misreporting agents not verified (see also \cite[Sec.~5]{FT10}). Specifically, a randomized mechanism $F$ is \emph{oblivious} if the probability distribution of $F$ over the outcomes, conditional on the event that no misreporting agents are included in the verification set, is identical to the probability distribution of $F$ if all misreporting agents are excluded from the mechanism. %Namely, misreporting agents not verified do not affect the allocation of $F$.
By induction on the number of agents, we show that obliviousness is a sufficient condition for robustness (Lemma~\ref{l:oblivious}).
To the best of our knowledge, this is the first time that robustness (or a similar) property is considered in mechanism design. We defer the discussion about robustness and its comparison to truthfulness to Section~\ref{s:app:concl}.
%
%Actually, with the possible exception of constant mechanisms, whose probability distribution over outcomes is independent of the agent declarations, a mechanism can be robust only if it uses exact verification.

Robustness leaves each agent with essentially two strategies: either she reports truthfully and participates in the mechanism or she lies and is excluded from the mechanism. An allocation rule satisfies \emph{voluntary participation} (or simply, \emph{participation}\,) if each agent's utility when she is truthful is no less than her utility when she is excluded from the mechanism.
%
%if each agent maximizes her utility by truthfully participating in the mechanism.
%
Robustness and participation immediately imply truthfulness%
\footnote{The reader is invited to verify that the average mechanism for facility location on the line is scale invariant, not strongly anonymous, oblivious (and thus, robust) and satisfies participation. Robustness and participation imply that the mechanism is truthful.}
(Lemma~\ref{l:robustness+participation}).
We prove that strongly anonymous randomized allocation rules that satisfy participation are closely related to \emph{maximal in distributional range} rules (see e.g., \cite{DD13,LS11}), i.e., allocation rules that maximize the expected social welfare over a (not necessarily proper) subset of probability distributions over outcomes.
Specifically, we show that maximizing the social welfare is sufficient for participation (Lemma~\ref{l:participation}), while for scale invariant and continuous allocation rules, it is also necessary (Lemma~\ref{l:chara_weak_parti}).
%
%Therefore, for strongly anonymous randomized allocation rules that satisfy participation, we can restrict our attention to maximal in distributional range allocations.

As a proof of concept, we apply selective verification to $k$-Facility Location problems (Section~\ref{s:facility}), which have served as benchmarks in approximate mechanism design without money (see e.g., \cite{PT09,AFPT09,LSWZ10,FT13} and the references therein). We show that Greedy (\cite[Section~2.2]{WS10}) and Proportional \cite{LSWZ10} satisfy participation and are robust and truthful, if we verify the $k$ agents allocated the facilities (Theorems~\ref{th:k-center}~and~\ref{th:k-facility-location}). %Greedy is known to be $2$-approximate for the objective of maximum cost and Proportional is $\Theta(\ln k)$-approximate for the objective of social cost.

For the general setting of utilitarian voting, we aim at strongly anonymous randomized allocation rules that are maximal in distributional range, so that they satisfy participation, and oblivious, so that they achieve robustness. %Participation and robustness imply that the mechanism with verification is truthful.
In Section~\ref{s:power}, we present the \emph{Power mechanism}, which selects each outcome $o$ with probability proportional to the $\ell$-th power of the total utility for $o$, where $\ell \geq 0$ is a parameter. Intuitively, Power provides a smooth transition from the (robust and truthful) uniform allocation, where each outcome is selected with probability $1/m$, for $\ell = 0$, to the optimal solution, for $\ell \rightarrow \infty$. Power approximately maximizes the social welfare and approximately satisfies participation. It is also scale invariant and, due to the proportional nature of its probability distribution, is oblivious and robust. Power can be implemented with selective verification of at most $\ell$ agents. Using $\ell =\ln m/\eps$, we obtain that for any $\eps > 0$, Power with selective verification of $\ln m/\eps$ agents, is robust, $\eps$-truthful and $(1-\eps)$-approximate for the social welfare (Theorem~\ref{th:power}).

To quantify the improvement, we show that without verification, in the general setting of utilitarian voting, the best possible approximation ratio of any randomized truthful mechanism is $1/m$ (see Section~\ref{s:app:gen_lower_bound}). In a slightly more restricted setting with injective valuations \cite{RM14}, the best known randomized truthful mechanism has an approximation ratio of $\Theta(m^{-3/4})$ and the best possible approximation ratio is $O(m^{-2/3})$. Moreover, the amount of verification is essentially best possible, since we prove that any truthful mechanism with constant approximation ratio needs to verify $\Omega(\log m)$ agents  (Theorem~\ref{th:lower_bound}). We essentially match this lower bound, that applies to all mechanisms, by strongly anonymous and scale invariant mechanisms.

In Section~\ref{s:impossibility}, we characterize the class of scale invariant and strongly anonymous truthful mechanisms that verify $o(n)$ agents and achieve \emph{full allocation}, i.e., they result in some outcome with probability $1$. We prove that any such mechanism must employ a constant allocation rule, i.e., a probability distribution that does not depend on the agent declarations. Therefore, such mechanisms cannot achieve nontrivial approximation guarantees. %In addition to justifying the fact that Power is almost truthful, 
Our characterization reveals an interesting and deep connection between continuity (which is necessary for low verification), full allocation, and maximal in distributional range mechanisms.

Relaxing some of the properties in the characterization, we can obtain (fully) truthful mechanisms with low verification. Relaxing full allocation, we obtain the \emph{Partial Power mechanism} (Section~\ref{s:partial}), and relaxing scale invariance, we obtain the \emph{Exponential mechanism} (Section~\ref{s:exponential}). Both are truthful, robust. For any $\eps > 0$, they verify $O(\ln m/\eps^2)$ agents in the worst-case and $\ln m/\eps$ agents in expectation, respectively. Partial Power is $(1-\eps)$-approximate, while Exponential has an additive error of $\eps n$. For Exponential, we can have an approximation ratio of $1-\eps$, given a constant factor estimation of the value maximum social welfare. All the mechanisms can be implemented in polynomial or expected polynomial time in $n$ and $m$.

\begin{figure}[t]
{\fontsize{10pt}{11.5pt}\selectfont\centerline{\begin{tabular}{ || c | c | c || }%
    \hline\hline
    \textbf{Power} & \textbf{Partial Power} & \textbf{Exponential}\\
    \hline\hline
    \textbf{$\eps$-truthful} & truthful & truthful\\
%   \hline
    full allocation & \textbf{partial allocation} & full allocation \\
%    \hline
    scale invariant & scale invariant & \textbf{not scale invariant}\\
%    \hline
    robust & robust & robust\\
%    \hline
\ \ \ $(1-\eps)$-approximation\ \ \  &\ \ \ $(1-\eps)$-approximation\ \ \ &
\ \ \ additive error $\eps n$\ \ \ \\
%    \hline
    verification $\ln m/\eps$ &
    \ \ \ verification $O(\ln m/\eps^2)$\ \ \ &
    \ \ \ expected verification $\ln m/\eps$\ \ \ \\
    \hline\hline
\end{tabular}}}
\vspace*{-0.2cm}\caption{\label{fig:comparison}The main properties of our mechanisms. Partial allocation means that the mechanism may result in an artificial outcome of valuation $0$ for all agents (e.g., we may refuse to allocate anything, for private goods, or to setup the service, for public goods). We depict in bold the property whose relaxation allows the mechanism to escape the characterization of Theorem~\ref{th:impossibility}.}
\vspace*{-0.6cm}\end{figure}

The properties of our mechanisms are summarized in Fig.~\ref{fig:comparison}. In all cases, we achieve a smooth tradeoff between the number of agents verified and the quality of approximation. Rather surprisingly, the verification depends on $m$, the number of outcomes, but not on $n$, the number $n$ of agents. Also, we discuss (Section~\ref{s:app:applications}) an application to the Combinatorial Public Project problem (see e.g., \cite{SS08,PSS08}).

\smallskip\noindent{\bf Related Work.}
Due to lack of space, we restrict our attention to the most relevant previous work (see also Section~\ref{s:app:previous}).
Previous work \cite{AK08,CESY12,FZ13} demonstrated that partial verification is essentially useless in the design of truthful mechanisms. Therefore, verification should be exact, i.e., it should forbid even negligible deviations from the truth, at least for some types of misreports. Thus, recent research has focused on the power of exact verification schemes that use either limited or costly verification and mild penalties.

In this direction, \cite{CESY12} introduces \emph{probabilistic verification} as a general framework for the use of verification in mechanism design.
%The idea is that any deviation from the truth is detectable with a probability depending on the distance of the misreport to the true type.
They show that almost any allocation rule can be implemented whit a truthful mechanism with money and probabilistic verification, provided that (i) the detection probability is positive for all agents and for negligible deviations from the truth; and that (ii) each liar incurs a sufficiently large penalty. Here, we instead use selective verification and the reasonable penalty of ignoring misreports and we verify only a small subset of agents instead of almost all of them.

%They show that any truthful mechanism with money can be simulated by a truthful mechanism without money and with probabilistic verification, provided that (i) the detection probability is positive for all agents and for negligible deviations from the truth; and that (ii) each liar incurs a sufficiently large penalty, which works essentially in the same way as monetary transfers. Here, we instead use selective verification and the reasonable penalty of ignoring misreports.

Our approach of selective verification is conceptually similar to the setting of \cite{BDL14}, which considers truthful allocation of an indivisible good without money and with \emph{costly selective verification} and seeks to maximize the social welfare minus the verification cost. Nevertheless, our setting and our mechanisms are much more general, we resort to approximate mechanisms (rather than exact ones) and treat the verification cost as a different efficiency criterion (instead of incorporating it in the social objective). Moreover, selective verification bears some resemblance to \cite{HR08}, which considers truthful mechanisms with money for single-unit and multi-unit auctions and aims at a good tradeoff between the social welfare and the payments charged.

There is a significant amount of work on mechanism design with verification where either the structure of the optimal mechanism is characterized (see e.g., \cite{SV14}), or mechanisms with money and verification are shown to achieve better approximation guarantees than mechanisms without verification (see e.g., \cite{ADPP09,KV10}). To the best of our knowledge, our work is the first where truthful mechanisms with limited selective verification (instead of partial or ``one-sided'' verification applied to all agents with positive utility) are shown to achieve best possible approximation guarantees for the general domain of utilitarian voting.

From a technical viewpoint, the idea of partial allocation in approximate mechanism design without money has been employed with remarkable success in \cite{CGG13}. However, this technique can achieve very restricted results in maximizing social welfare without verification in a very general setting as utilitarian voting (see Section~\ref{s:app:gen_lower_bound}). Moreover, our motivation for using the exponential mechanism with selective verification came from \cite{NST10,HK12}, due their tradeoffs between the approximation guarantee and the probability of the gap mechanism (resp. amount of payments) required for truthfulness.

\clearpage

%% file: model.tex
\section{Notation and Preliminaries}
\label{s:prelim}

For any integer $m \geq 1$, we let $[m] \equiv \{ 1, \ldots, m \}$. For an event $E$, $\Prob[E]$ denotes the probability of $E$. For a random variable $X$, $\Exp[X]$ denotes the expectation of $X$. For a finite set $S$, $\Delta(S)$ is the unit simplex over $S$, which includes all probability distributions over $S$.
For a vector $\vec{x} = (x_1, \ldots, x_m)$ and some $j \in [m]$, $\vec{x}_{-j}$ is $\vec{x}$ without $x_j$. For a nonempty $S \subseteq [m]$, $\vec{x}_S = (x_j)_{j\in S}$ is the projection of $\vec{x}$ to $S$.
For vectors $\vec{x}$ and $\vec{y}$, $\vec{x}+\vec{y} = (x_1+y_1, \ldots, x_m+y_m)$ denotes their coordinate-wise sum.
%
%This notation extends naturally to the summation of any finite number of vectors.
%
For a vector $\vec{x}$ and an $\ell \geq 0$, $\vec{x}^{\ell} = (x_1^\ell, \ldots, x_m^\ell)$ is the coordinate-wise power of $\vec{x}$ and
$\norm{\vec{x}}_\ell = ( \sum_{j=1}^m x_j^\ell)^{1/\ell}$ is the $\ell$-norm of $\vec{x}$. For convenience, we let $\norm{\vec{x}}_1 = | \vec{x} |$. Moreover, $\norm{\vec{x}}_\infty = \max_{j\in[m]}\{ x_j \}$ is the infinity norm of $\vec{x}$.

\smallskip\noindent{\bf Agent Valuations.}
We consider a set $N$ of $n$ strategic agents with private preferences over a set $O$ of outcomes. We focus on combinatorial problems, assume that $O$ is finite and let $m \equiv |O|$ be the number of different outcomes.
The preferences of each agent $i$ are given by a \emph{valuation function} or \emph{type} $\vec{x}_i : O \to \reals_{\geq 0}$ that $i$ seeks to maximize. The set of possible valuations is the \emph{domain} $D = \reals_{\geq 0}^m$. We usually regard each valuation as a vector $\vec{x}_i = (x_i(j))_{j\in[m]}$, where $x_i(j)$ is $i$'s valuation for outcome $j$.
A valuation profile is a tuple $\matr{x} = (\vec{x}_1, \ldots, \vec{x}_n)$ consisting of the agents' valuations.
Given a valuation profile $\matr{x}$, $\vec{w}(\matr{x}) = \vec{x}_1+\cdots+\vec{x}_n$ is the vector of the total valuation, or simply, of the \emph{weight}, for each outcome. We usually write $\vec{w}$, instead of $\vec{w}(\matr{x})$, when $\matr{x}$ is clear from the context.

\smallskip\noindent{\bf Allocation Rules.}
A (randomized) \emph{allocation rule} $f : D^n \to \Delta(O)$ maps each valuation profile to a probability distribution over $O$. To allow for exclusion of some agents from $f$, we always assume that $f$ is well defined for any number of agents $n'$, $0 \leq n' \leq n$.
%
%Formally, we let $f : \union_{n'=0}^n D^{n'} \to \Delta(O)$.
%
We regard the probability distribution of $f$ on input $\matr{x}$ as a vector $f(\matr{x}) = (f_{j}(\matr{x}))_{j\in[m]}$, where $f_j(\matr{x})$ is the probability of outcome $j$.
Then, the expected utility of agent $i$ from $f(\matr{x})$ is equal to the dot product $\vec{x}_i \cdot f(\matr{x})$.
An allocation rule is \emph{constant} if for all valuation profiles $\vec{x}$ and $\vec{y}$, $f(\vec{x}) = f(\vec{y})$, i.e., the probability distribution of $f$ in independent of the valuation profile. E.g., the uniform allocation rule, that selects each outcome with probability $1/m$, is constant.

A rule $f$ achieves \emph{full allocation} if for all $\matr{x}$, $|f(\matr{x})| = 1$, and \emph{partial allocation} if $|f(\matr{x})| < 1$, for some $\matr{x}$. A full allocation rule always outputs an outcome $o \in O$, while a partial allocation rule may also output an artificial (or \emph{null}) outcome not in $O$. We assume that all agents have valuation $0$ for the null outcome.

 Two nice properties of our allocation rules is that they are strongly anonymous and (most of them) scale invariant.
An allocation rule $f$ is \emph{scale invariant} if for any valuation profile $\matr{x}$ and any $\alpha \in \reals_{> 0}$, $f(\alpha\matr{x}) = f(\matr{x})$, i.e., scaling all valuations in $\matr{x}$ by $\alpha$ does not change the allocation.
An allocation rule $f$ is \emph{strongly anonymous} if $f(\matr{x})$ depends only on the vector $\vec{w}(\matr{x})$ with outcome weights. Formally, for all valuation profiles $\matr{x}$ and $\matr{y}$ (possibly with a different number of agents) with $\vec{w}(\matr{x}) = \vec{w}(\matr{y})$, $f(\matr{x}) = f(\matr{y})$.
Hence, a strongly anonymous rule can be regarded as a one-agent allocation rule $f : D \to \Delta(O)$.
%
%Then, a valuations profile $\matr{x}$ can be regarded as an $m$-dimensional vector whose $j$-th coordinate is the total agents' valuation for $j$, i.e., $\vec{x} =\sum_{i=1}^n \vec{x}_i$.
%
%In the following sections, all allocation rules (and mechanisms) will be strongly anonymous, unless stated otherwise.

\smallskip\noindent{\bf Approximation Guarantee.}
The social efficiency of an allocation rule $f$ is evaluated by a \emph{social objective function} $g : D^n \times O \to \reals_{\geq 0}$.
We mostly consider the objective of \emph{social welfare}, where we seek to maximize $\sum_{i=1}^n \vec{x}_i \cdot f(\matr{x})$.
The optimal social welfare of a valuation profile $\matr{x}$ is $\| \sum_{i=1}^n \vec{x}_i \|_\infty$\,.
An allocation rule $f$ has \emph{approximation ratio} $\rho \in (0, 1]$ (resp. \emph{additive error} $\delta > 0$) if for all valuation profiles $\matr{x}$,
\( \sum_{i=1}^n \vec{x}_i \cdot f(\matr{x}) \geq \rho \left\| \sum_{i=1}^n \vec{x}_i \right\|_\infty \)
(resp.
\( \sum_{i=1}^n \vec{x}_i \cdot f(\matr{x}) \geq \left\| \sum_{i=1}^n \vec{x}_i \right\|_\infty - \delta \)
).

\smallskip\noindent{\bf Voluntary Participation and MIDR.}
An allocation rule $f$ satisfies \emph{voluntary participation} (or simply, \emph{participation}) if for any agent $i$ and any valuation profile $\matr{x}$,
\( \vec{x}_i \cdot f(\matr{x}) \geq \vec{x}_i \cdot f(\matr{x}_{-i})
\),
i.e., $i$'s utility does not decrease if she participates in the mechanism.
For some $\e \in (0, 1]$, $f$ satisfies \emph{$\e$-participation} if for all agents $i$ and valuation profiles $\matr{x}$,
\( \vec{x}_i \cdot f(\matr{x}) \geq \e\,\vec{x}_i \cdot f(\matr{x}_{-i}) \).
An allocation rule $f$ is \emph{maximal in distributional range} (MIDR) if there exist a range $Z$ of (possibly partial) allocations and a function $h : Z \to \reals$ such that for all valuation profiles $\vec{x}$,
$f(\matr{x}) = \arg\max_{\vec{z} \in Z} \sum_{i=1}^n \vec{x}_i \cdot \vec{z} + h(\vec{z})$ (see e.g., \cite{DD13,LS11}).
The following show that MIDR is a sufficient condition for participation and that for scale invariant and strongly anonymous continuous allocations, MIDR is also necessary (the proofs can be found in Section~\ref{s:app:participation} and Section~\ref{s:app:chara_weak_parti}).

\begin{lemma}\label{l:participation}
Let $f$ be any MIDR allocation rule. Then, $f$ satisfies participation.
\end{lemma}

\begin{lemma}\label{l:chara_weak_parti}
For any scale invariant and strongly anonymous continuous allocation rule $f$ that satisfies participation, there is a range $Z$ of (possibly partial) allocations such that $f(\vec{x}) = \arg \max_{\vec{z} \in Z} \vec{x} \cdot \vec{z}$.
\end{lemma}

\section{Mechanisms with Selective Verification and Basic Properties}
\label{s:model}

A \emph{mechanism with selective verification} $F$ takes as input a reported valuation profile $\matr{y}$ and has \emph{oracle access} to a binary \emph{verification vector} $\vec{s} \in \{ 0, 1 \}^n$, with $s_i = 1$ if agent $i$ has truthfully reported $\vec{y}_i = \vec{x}_i$, and $s_i = 0$ otherwise. In fact, we assume that $F$ verifies an agent $i$ through a \emph{verification oracle} $\ver$ that on input $i$, returns $\ver(i) = s_i$. So, we regard a mechanism with verification as a function $F : D^n \times \{ 0, 1 \}^n \to \Delta(O)$. We highlight that although the entire vector $\vec{s}$ appears as a parameter of $F$, for notational convenience, the outcome of $F$ actually depends on few selected coordinates of $\vec{s}$.
We denote $V(\matr{y}) \subseteq N$, or simply $V$, the set of agents verified by $F$ on input $\matr{y}$.
As for allocation rules, we treat the probability distribution of $F$ over outcomes as an $m$-dimensional vector and assume that $F$ is well defined for any number of agents $n' \leq n$.

Our approach is to start from an allocation rule $f$ and to devise a mechanism $F$ that motivates truthful reporting by selective verification. We say that a mechanism $F$ with selective verification is \emph{recursive} if there is an allocation rule $f$ such that $F$ operates as follows: on a valuation profile $\matr{y}$, $F$ selects an outcome $o$, with probability $f_o(\vec{y})$, and a verification set $V(\matr{y})$, and computes the set $L = \{ i \in V(\matr{y}) : \ver(i) = 0 \}$ of misreporting agents in $V(\matr{y})$. If $L = \emptyset$, $F$ returns $o$. Otherwise, $F$ recurses on $\matr{y}_{-L}$. Our mechanisms are recursive, except for Partial Power (Section~\ref{s:partial}), which adopts a slightly different reaction to $L \neq \emptyset$.

Given an allocation rule $f$, we say that a mechanism with verification $F$ is an \emph{extension} of $f$ is for all valuation profiles $\matr{x}$, $F(\matr{x}, \vec{1}) = f(\matr{x})$. Namely, $F$ behaves exactly as $f$ given that all agents report truthfully. For the converse, given a mechanism $F$, we say that $F$ \emph{induces} an allocation rule $f$ if for all $\matr{x}$, $f(\matr{x}) = F(\matr{x}, \vec{1})$. For clarity, we refer to mechanisms with selective verification simply as \emph{mechanisms}, and denote them by uppercase letters, and to allocation rules simply as \emph{rules} or \emph{algorithms}, and denote them by lowercase letters. A mechanism $F$ has a property of an allocation rule (e.g., scale invariance, partial or full allocation, participation, approximation ratio) iff the induced rule $f$ has this property.

A mechanism $F$ is \emph{$\e$-truthful}, for some $\e \in (0, 1]$, if for any agent $i$, for any valuation pair $\vec{x}_i$ and $\vec{y}_i$ and for all reported valuations $\matr{y}_{-i}$ and verification vectors $\vec{s}_{-i}$,
\[ \vec{x}_i \cdot F((\matr{y}_{-i}, \vec{x}_i), (\vec{s}_{-i}, 1)) \geq
   \e\, \vec{x}_i \cdot F((\matr{y}_{-i}, \vec{y}_i), (\vec{s}_{-i}, 0))
\]
A mechanism $F$ is \emph{truthful} if it is $1$-truthful. Namely, no matter the reported valuations of the other agents and whether they report truthfully or not, the expected utility of agent $i$ is maximized if she reports truthfully.

\smallskip\noindent{\bf Robustness and Obliviousness.}
A remarkable property of our mechanisms is \emph{robustness}, namely that they ignore the valuations of misreporting agents and let their outcome be determined by the valuations of truthful agents only.
Formally, a mechanism $F$ is \emph{robust} if for all reported valuations $\matr{y}$ and verification vectors $\vec{s}$,
$F(\matr{y}, \vec{s}) = F(\matr{y}_{T(\vec{s})}, (1, \ldots 1))$, with the equality referring to the probability distribution of $F$,
where $T(\vec{s}) = \{ i \in N: s_i = 1 \}$ is the set of truthful agents. Next, we simply use $T$, instead of $T(\vec{s})$.

A mechanism with selective verification $F$ is \emph{oblivious} (to the declarations of misreporting agents not verified) if for all valuation profiles $\matr{y}$ and verification vectors $\vec{s}$, with $L = N \setminus T(\vec{s})$, and any outcome $o$,
\begin{equation}\label{eq:oblivious}
   \Prob[ F(\matr{y}, \vec{s}) = o\,|\,V(\matr{y}) \cut L = \emptyset] =
   \Prob[ F(\matr{y}_{-L}, \vec{1}) = o]
\end{equation}
I.e., if the misreporting agents are not caught, they do not affect the probability distribution of $F$ (see also \cite{FT10}). By induction on the number of misreports, we show that obliviousness is sufficient for robustness.

\begin{lemma}\label{l:oblivious}
Let $F$ be any oblivious recursive mechanism with selective verification. Then, $F$ is robust.
\end{lemma}

\begin{proof}
We fix a valuation profile $\matr{y}$ and a verification vector $\vec{s}$, and prove that for any outcome $o \in O$, $\Prob[F(\matr{y}, \vec{s}) = o] = \Prob[F(\matr{y}_{T}, \vec{1}) = o]$, where $T \subseteq N$ is the set of truthful agents in $\matr{y}$. Clearly, this implies that $F$ is robust. The proof is by induction on the number of agents $N$.

If $N = \emptyset$, the statement is obvious. So, we assume inductively that the statement holds for every proper subset of $N$. Let $L = N \setminus T$ be the set of misreporting agents in $\matr{y}$ and let $V$ be the verification set of $F$ on input $\matr{y}$. Then,
\begin{equation}\label{eq:robust}
 \Prob[F(\matr{y}, \vec{s}) = o] =
 \sum_{L' \subseteq L} \Prob[F(\matr{y}, \vec{s}) = o\,|\,V \cut L = L']
 \,\Prob[V \cut L = L']
\end{equation}
We have that
\(
   \Prob[F(\matr{y}, \vec{s}) = o\,|\,V \cut L = \emptyset] =
%   \Prob[F(\matr{y}_{-L}, \vec{s}_{-L}) = o] =
   \Prob[F(\matr{y}_{T}, \vec{1}) = o]
\),
by (\ref{eq:oblivious}), since $F$ is oblivious.
%and the agents in $T = N \setminus L$ are truthful in $\matr{y}$.
If $V$ includes a non-empty set $L' = V \cut L$, since $F$ is recursive, it ignores their declarations and recurses on $\matr{y}_{-L'}$. Therefore, for all $\emptyset \neq L' \subseteq L$,
\(
   \Prob[F(\matr{y}, \vec{s}) = o\,|\,V \cut L = L'] =
   \Prob[F(\matr{y}_{-L'}, \vec{s}_{-L'}) = o] =
   \Prob[F(\matr{y}_{T}, \vec{1}) = o]
\)\,,
where the last equality follows from the induction hypothesis, because the agents in $\matr{y}_{-L'}$ are a proper subset of agents in $\matr{y}$.
Therefore, using that $\Prob[F(\matr{y}, \vec{s}) = o\,|\,V \cut L = L'] = \Prob[F(\matr{y}_{T}, \vec{1}) = o]$, for all $L' \subseteq L$, in (\ref{eq:robust}), we obtain that  $\Prob[F(\matr{y}, \vec{s}) = o] = \Prob[F(\matr{y}_{T}, \vec{1}) = o]$, i.e., that $F$ is robust.
\qed\end{proof}

\noindent{\bf Robustness, Participation and Truthfulness.}
In the Appendix, Section~\ref{s:app:robustness+participation}, we show that robustness and participation imply truthfulness (note that the converse may not be true, since a truthful mechanism with verification does not need to be robust). Then, by Lemma~\ref{l:participation} and Lemma~\ref{l:oblivious}, we can focus on MIDR allocation rules for which the outcome and the verification set can be selected in an oblivious way.

\begin{lemma}\label{l:robustness+participation}
For any $\e \in (0, 1]$, if a mechanism with selective verification $F$ is robust and satisfies $\e$-participation, then $F$ is $\e$-truthful.
\end{lemma}

\noindent{\bf Quantifying Verification.}
Focusing on truthful mechanisms with verification, where the agents do not have any incentive to misreport, we bound the amount of verification when the agents are truthful (similarly to the definition of the approximation ratio of $F$ as the approximation ratio of the induced allocation rule $f$). For a truthful mechanism $F$, this is exactly the amount of verification required so that $F$ motivates truthfulness.

Given a mechanism with verification $F$, its \emph{worst-case verification} is
$\Ver(F) \equiv \max_{\matr{x} \in D^n} |V(\vec{x})|$, i.e., the maximum number of agents verified by $F$ in any truthful valuation profile. 
%If $F$ is randomized, the maximum is taken also over all possible random strings used by $F$. 
If $F$ is randomized, its \emph{expected verification} is $\ExpVer(F) \equiv \max_{\matr{x} \in D^n} \Exp[|V(\matr{x})|]$, where expectation is over all random strings used.

%% file: facility.tex
\section{Motivating Example: Facility Location Mechanisms with Selective Verification}
\label{s:facility}

As a proof of concept, we apply mechanisms with verification to $k$-Facility Location. In such problems, we have a metric space $(M, d)$, where $M$ is a finite set of points and $d$ is a metric distance function.
The outcomes are all subsets of $k$ locations in $M$. Each agent $i$ has a preferred location $t_i \in M$ and her ``valuation'' for outcome $C$ is $\vec{x}_i(C) = -d(t_i, C)$, i.e., minus the distance of her preferred location to the nearest facility in $C$. %The minus sign is due to the cost minimization nature of the problem. 
So, each agent $i$ aims at minimizing $d(t_i, C)$.
The mechanism $F$ gets a profile $\vec{z} = (z_1, \ldots, z_n)$ of reported locations. Using access to a verification oracle, $F$ maps $\vec{z}$ to a set $C$ of $k$ facility locations.

\smallskip\noindent{\bf Maximum Cost.}
To minimize $\max_{i \in N} \{ d(t_i, F(\vec{t}, \vec{1})) \}$, i.e, the maximum agent-facility distance, we use the $2$-approximate Greedy algorithm for $k$-Center (see e.g., \cite[Sec.~2.2]{WS10}). On input $\vec{z}$, Greedy first allocates a facility to an arbitrary agent. As long as $|C| < k$, the next facility is allocated to the agent $i$ maximizing $d(z_i, C)$. We extend Greedy to a mechanism with selective verification by inspecting the reported location $z_i$ of every agent $i$ allocated a facility. If all of them are truthful, we place the $k$ facilities at $C$. Otherwise, we exclude any liars in $C$ and recurse on the remaining agents. In the Appendix, Section~\ref{s:app:k-center}, we establish the properties of Greedy with verification. To quantify the improvement due to the use of verification, we highlight that there are no deterministic truthful mechanisms (without verification) that place $k \geq 2$ facilities in tree metrics and achieve a bounded (in terms of $n$ and $k$) approximation ratio (see \cite{FT12}).

\begin{theorem}\label{th:k-center}
The Greedy mechanism with verification for $k$-Facility Location is truthful and robust, is $2$-approximate for the maximum cost and verifies $k$ agents.
\end{theorem}

\noindent{\bf Social Cost.}
To minimize $\sum_{i=1}^n d(t_i, F(\vec{t}, \vec{1}))$, i.e, the total cost of the agents, we use the Proportional mechanism \cite{LSWZ10}, which is $\Theta(\ln k)$-approximate \cite{AV07}. Proportional first allocates a facility to an agent chosen uniformly at random. As long as $|C| < k$, agent $i$ is allocated the next facility with probability proportional to $d(z_i, C)$. Verifying the reported location of every agent allocated a facility, we obtain that (see Section~\ref{s:app:k-facility-location}):

\begin{theorem}\label{th:k-facility-location}
The Proportional mechanism with verification for $k$-Facility Location is truthful and robust, is $\Theta(\ln k)$-approximate for the social cost and verifies $k$ agents.
\end{theorem}

%Both mechanisms have good approximation guarantees and motivate truthfulness by verifying only the $k$ agents who get the facilities. Most important, the proof that the mechanisms are truthful and robust only assumes that the cost of agent $i$ is minimized if $i$ gets a facility at her true location. Hence, the mechanisms remain robust and truthful even if the distances do not satisfy the triangle inequality and the cost function of each agent $i$ is any function of $d(t_i, C)$ (as long as this function is minimized at $0$).
%Robustness is guaranteed independently of the cost function.
%Of course, their approximation guarantees only hold for metric distances and linear agent costs.

%% file: power.tex
\section{The Power Mechanism with Selective Verification}
\label{s:power}

\begin{algorithm}[t]
\caption{\label{alg:power}The Power Mechanism $\Pow^\ell(\matr{x}, \vec{s})$}
\begin{algorithmic}\normalsize
    \State let $N$ be the set of the remaining agents and let $L \leftarrow \emptyset$
    \State pick an outcome $j \in O$ and a tuple $\vec{t} \in N^{\ell}$ \\
\ \ \ \ \ \ with probability proportional to the value of the
            term $x_{t_1}(j) x_{t_2}(j) \cdots x_{t_\ell}(j)$
    \For{\textbf{each} agent $i \in \vec{t}$}
        \If{$\ver(i) \neq 1$} $L \leftarrow L \union \{ i \}$ \EndIf
    \EndFor
    \If{$L \neq \emptyset$} \Return $\Pow^\ell(\matr{x}_{-L}, \vec{s}_{-L})$
    \Else\ \Return outcome $j$
    \EndIf
\end{algorithmic}\end{algorithm}%\vspace*{-6mm}

In this section, we present the Power mechanism, a recursive mechanism with verification that approximates the social welfare in the general domain of utilitarian voting. Power with parameter $\ell \geq 0$ (or $\Pow^\ell$, for brevity, see also 
Mechanism~\ref{alg:power}) is based on a strongly anonymous and scale invariant allocation that assigns probability proportional to the weight of each outcome raised to $\ell$. Hence, for each valuation profile $\matr{x}$, the outcome of $\Pow^\ell$ depends on the weight vector $\vec{w} = \sum_{i=1}^n \vec{x}_i$. If all agents are truthful, $\Pow^\ell$ results in each outcome $j$ with probability
$w_j^{\ell} / \sum_{q=1}^m w_q^{\ell}$, i.e., proportional to $w_j^{\ell}$
(note that for $\ell = 0$, we get the uniform allocation, while for $\ell = \infty$, the outcome of maximum weight gets probability $1)$.
To implement this allocation with low verification, we observe that each term $w_j^{\ell}$ can be expanded in $n^{\ell}$ terms as follows%
\footnote{For example, let $n = 3$, $\ell = 2$ and $w_j = x_1 +x_2+x_3$ (we omit $j$ from $x$'s for clarity). In (\ref{eq:terms}), we expand $w_j^2$ in $3^2 = 9$ terms as follows $w_j^2 = (x_1 +x_2+x_3)^2 = x_1x_1+x_1x_2+x_1x_3+x_2x_1+x_2x_2+x_2x_2+x_3x_1+x_3x_2+x_3x_3$. Hence, in this example, $\vec{t} \in \{1,2,3\}\times\{1,2,3\}$. Given that outcome $j$ is chosen, each of these terms (and the corresponding tuple $\vec{t}$) is selected with probability proportional to its value. E.g., $x_1x_2$ and $\vec{t} = (1, 2)$ are selected with probability $x_1x_2 / w^2_j$.}:
\begin{equation}\label{eq:terms}
 w_j^{\ell} =
 \left( \sum_{i \in N} x_i(j) \right)^{\!\!\ell} =
 \sum_{\vec{t} \in N^{\ell}}
 x_{t_1}(j) x_{t_2}(j) \cdots x_{t_\ell}(j)
\end{equation}
Hence, choosing an outcome $j$ and a tuple $\vec{t} \in N^{\ell}$ with probability proportional%
\footnote{\label{foot:sampling}To sample from (\ref{eq:terms}) in $O(m + n \ell)$ steps, we select outcome $j$ with probability $w^\ell_j / |\vec{w}^{\ell}|$, and then, conditional on $j$, we select agent $i$ in each position of $\vec{t}$ independently with probability $x_i(j) / w_j$. Each tuple $\vec{t}$ is picked with probability $x_{t_1}(j) \cdots x_{t_\ell}(j)/|\vec{w}^{\ell}|$.}
to $x_{t_1}(j) x_{t_2}(j) \cdots x_{t_\ell}(j)$, we end up with each outcome $j$ with probability proportional to $w_j^{\ell}$.
The verification set of $\Pow^\ell$ consists of the agents in $\vec{t}$. Since at most $\ell$ agents contribute to each term $x_{t_1}(j) x_{t_2}(j) \cdots x_{t_\ell}(j)$, we can make $\Pow^\ell$ robust and almost truthful by verifying at most $\ell$ agents.
In Section~\ref{s:app:power}, we show that $\Pow^\ell$ is oblivious, due to its proportional nature, and thus robust, and satisfies $m^{-1/(\ell+1)}$-participation. Thus, we obtain that:

\begin{theorem}\label{th:power}
For any $\eps > 0$, $\Pow^\ell$ with $\ell = \ln m / \eps$ is robust and $(1-\eps)$-truthful, has worst-case verification $\ln m/\eps$, and achieves an approximation ratio of $(1-\eps)$ for the social welfare. 
\end{theorem}

%% file: lowerbound.tex
\section{Logarithmic Verification is Best Possible}
\label{s:lower_bound}

Next, we describe a random family of instances where truthfulness requires a logarithmic expected verification.
%Using a simple probabilistic method argument we can then conclude that there exists an instance where the worst case verification is logarithmic.
Below, we only sketch the main idea of the proof. The full proof can be found in Section~\ref{s:app:lower_bound}.

\begin{theorem}\label{th:lower_bound}
Let $F$ be randomized truthful mechanism that achieves a constant approximation ratio for any number of agents $n$ and any number of outcomes $F$. Then, $F$ needs expected verification $\Omega(\log m)$.
\end{theorem}

\begin{proofsketch}
We consider $m$ outcomes and $m$ disjoint groups of agents. Each group has a large number $\nu$ of agents. An agent in group $j$ has valuation $0$ for any other outcome and valuation either $1$ or $\delta$ for outcome $j$, where $\delta > 0$ is tiny (e.g., $\delta = 1/\nu^{10}$). In each group $j$, the probability that $k$ agents, $0 \leq k \leq \nu$, have valuation $1$ for outcome $j$ is $2^{-(k+1)}$. The expected maximum social welfare of such instances is $\Theta(\log m)$.

We next focus on a group $j$ of agents and fix $\vec{x}_{-j}$, i.e., the declarations of the agents in all other groups. Using a simple argument, we can assume wlog. that the probability of outcome $j$ depends only on the number of agents in group $j$ that declare $1$ for $j$. Thus, the mechanism induces a sequence of probabilities $p_0, p_1, \ldots, p_k, \cdots$, where $p_k$ is the probability of outcome $j$, given that the number of agents that declare $1$ for $j$ is $k$.
Since the mechanism is truthful, if $k$ agents declare $1$ for outcome $j$, we need to verify each of them with probability at least $p_k - p_{k-1}$. Otherwise, an agent with valuation $\delta$ can declare $1$ and improve her expected utility.
Therefore, for any fixed $\vec{x}_{-j}$, when $k$ agents declare $1$ for outcome $j$, we need an expected verification of at least $k (p_k - p_{k-1})$ for agents in group $j$.

Assuming truthful reporting and taking the expectation over the number of agents in group $j$ with valuation $1$, we find that expected verification for agents in group $j$ is at least half the expected social welfare of the mechanism from group $j$, conditional on $\vec{x}_{-j}$, minus half the probability of outcome $j$, conditional on $\vec{x}_{-j}$. Removing the conditioning on $\vec{x}_{-j}$ and summing up over all groups $j$, we find that expected verification is at least half the expected welfare of the mechanism minus $1/2$. Since the mechanism has a constant approximation ratio, there are instances where the expected verification is $\Omega(\log m)$.
\qed\end{proofsketch}

%% file: impossibility.tex
\section{Characterization of Strongly Anonymous Mechanisms}
\label{s:impossibility}

Next, we characterize the class of scale invariant and strongly anonymous truthful mechanisms that verify $o(n)$ agents. The characterization is technically involved and consists of four main steps. We first prove that these rules are continuous (for full proof see Section~\ref{s:app:discotruth}).

\begin{lemma}\label{l:discotruth}
Let $f$ be any scale invariant and strongly anonymous allocation rule. If $f$ is discontinuous, every truthful extension $F$ of $f$ needs to verify $\Omega(n)$ agents in expectation, for arbitrarily large $n$.
\end{lemma}

\begin{proofsketch}
First, we prove that if $f$ has a discontinuity, there are $\Omega(n)$ agents that have a very small valuation $\delta > 0$ and can change the allocation by a constant factor, independent of $n$ and $\delta$. Next, we focus on any truthful extension $F$ of $f$ and show that for every agent $i$ that has the ability to change the allocation by a constant factor, the probability that $F$ verifies $i$ should be at least a constant, say $\zeta$, due to truthfulness. Therefore, the expected verification of $F$ is at least $\zeta \times \Omega(n) = \Omega(n)$.
\qed\end{proofsketch}

Therefore, if a truthful mechanism $F$ verifies $o(n)$ agents and induces a scale invariant and strongly anonymous allocation rule $f$, then $f$ needs to be continuous. In Section~\ref{s:app:cont-parti}, we prove that such an allocation rule $f$ satisfies participation. Then, by Lemma~\ref{l:chara_weak_parti}, we obtain the characterization that such an allocation rule $f$ is MIDR. Finally, in Section~\ref{s:app:discont}, we show that any full allocation and MIDR rule $f$ is either constant, i.e., its probability distribution does not depend on the valuation profile $\vec{x}$, or has a discontinuity at $\vec{1}$. Thus, we obtain the following characterization:

\begin{theorem}\label{th:impossibility}
Let $F$ be any truthful mechanism that verifies $o(n)$ agents, is scale invariant and strongly anonymous and achieves full allocation. Then, $F$ induces a constant allocation rule.
\end{theorem}

%% file: partial.tex
\begin{algorithm}[t]
\caption{\label{alg:partial_power}The Partial Power Mechanism $\PPow^{\ell, r}(\matr{x}, \vec{s})$}
  \begin{algorithmic}[1]\normalsize
      \State pick $r$ tuples $\vec{t}^{(1)},...,\vec{t}^{(r)} \in N^{\ell + 1}$ \\
      \ \ \ \ \ \ with probability proportional to the value of the term
      $\sum_{j \in O} x_{t_1}(j) x_{t_2}(j) \cdots x_{t_{\ell + 1}}(j)$

      \For{\textbf{each} $k \in \{1,...,r\}$ and agent $i \in \vec{t}^{(k)}$}
        \If{$\ver(i) \neq 1$ } {\Return $\bot$}  \EndIf
	\EndFor\smallskip
  \State {with} probability $1 - \sum_j f^{(\ell, r)}_j(\vec w)$ \Return {null}
      \State pick an outcome $j \in O$ and a tuple $\vec{t} \in N^{\ell}$ \\
      \ \ \ \ with probability proportional to the value of the term $x_{t_1}(j) x_{t_2}(j) \cdots x_{t_\ell}(j)$
      \For{\textbf{each} agent $i \in \vec{t}$}
        \If{$\ver(i) \neq 1$ } {\Return $\bot$}  \EndIf
      \EndFor
     \State  \Return outcome $j$
  \end{algorithmic}
\end{algorithm}

\section{The Partial Power Mechanism with Selective Verification}
\label{s:partial}

The Power mechanism, in Section~\ref{s:power}, escapes the characterization of Theorem~\ref{th:impossibility} by relaxing participation (and thus, truthfulness). In this section, we present Partial Power which escapes the characterization by relaxing full allocation. Thus, Partial Power results in some outcome in $O$ with probability less than $1$, and with the remaining probability, it results in an artificial \emph{null} outcome for which all agents have valuation $0$.

Lemma~\ref{l:chara_weak_parti} implies that social welfare maximization is essentially necessary for participation. The proof of Theorem~\ref{th:impossibility} implies that maximizing the social welfare over $\Delta(O)$ results in discontinuous mechanisms that need $\Omega(n)$ verification (e.g., let $m = 2$ and consider welfare maximization for weights $(1, 1+\e)$ and $(1, 1-\e)$, see also Lemma~\ref{l:discont}). Hence, we optimize over a smooth surface that is close to $\Delta(O)$, but slightly curved towards the corners, so that the resulting welfare maximizers are continuous.
Precisely,
%
%looking for a strongly anonymous allocation rule satisfying
%
%$f(\vec{w}) = \arg \max_{\vec{z} \in Z} \vec{w} \cdot \vec{z}$,
%
%for all weight vectors $\vec{w} \in \reals^m_{\geq 0}$,
we consider welfare maximization over the family of sets
$Z_{\ell, r} = \left\{\vec{z} \in \reals^m_{\geq 0} : \norm{\vec{z}}_{1+1/\ell} \leq (1 - 1/r) m^{-1/(\ell+1)} \right\}$
for all integers $\ell , r \geq 1$.
Welfare maximization over $Z_{\ell, r}$ results in
$f^{(\ell,r)}(\vec{w}) =  (1 - 1/r) \vec{w}^{\ell}/ \left( {m^{1/(\ell+1)} \norm{\vec{w}^{\ell}}_{1+1/\ell}}  \right)$ (Lemma~\ref{l:optimization}), a continuous allocation that is MIDR and satisfies participation.
Lemma~\ref{l:appro_ppower} shows that for any $\ell \geq 1$, the partial allocation $f^{(\ell,r)}$ has approximation ratio $(1-1/r) m^{-1 / (\ell + 1)}$ for the social welfare.

We next show that there exists a robust extension $\PPow^{\ell,r}$ of the allocation $f^{(\ell,r)}$ with reasonable verification. Thus, we establish that $\PPow^{\ell,r}$ is truthful. To this end, we introduce Mechanism~\ref{alg:partial_power}. Since $f^{(\ell,r)}$ is strongly anonymous, we consider below the weights $\vec{w}\equiv \vec{w}(\matr{x})$ instead of the valuations $\matr{x}$. If all agents are truthful, $\PPow^{\ell,r}$ samples exactly from $f^{(\ell,r)}(\vec{w})$. In particular, steps~1-4 never result in $\bot$, step~5 outputs null with probability $1-|f^{(\ell,r)}(\vec{w})|$, and steps~6-10 work identically to $\Pow^\ell$, since given that the null outcome is not selected, each outcome $j$ is chosen with probability proportional to $w_j^\ell$.

The most interesting case is when some agents misreport their valuations. To achieve robustness, we need to ensure that the probability distribution is identical to the case where misreporting agents are excluded from the mechanism. Similarly to $\Pow^\ell$, misreporting agents cannot affect the relative probabilities of each outcome. In $\PPow^{\ell,r}$ however, they may affect the probability of the null outcome. Thus, $\PPow^{\ell,r}$ is not oblivious and we cannot establish robustness through Lemma~\ref{l:oblivious} or some variant of it.

Robustness of $\PPow^{\ell,r}$ is obtained through the special action $\bot$, triggered when verification reveals some misreporting agents. Then, $\PPow^{\ell,r}$ needs to allocate appropriate probabilities to each outcome $j$ and to the null outcome so that the unconditional probability distribution of $\PPow^{\ell,r}$ is identical to $f^{(\ell,r)}(\vec{w}_T)$, where $T$ is the set of truthful agents. Therefore, whenever $\PPow^{\ell,r}$ returns $\bot$, we verify all agents, compute the weight vector $\vec{w}_T$ for the truthful agents, and return each outcome $j$ with
probability:
\[
 p_j = \frac{f^{(\ell,r)}_j(\vec{w}_T) - \Prob[\PPow^{\ell,r}(\vec{x}, \vec{s}) = j\,|\,\PPow^{\ell,r}(\vec{x}, \vec{s}) \neq \bot] \, \Prob[\PPow^{\ell,r}(\vec{x}, \vec{s}) \neq \bot]}{\Prob[\PPow^{\ell,r}(\vec{x}, \vec{s}) = \bot]}
\]

The null outcome is return with probability $1-\sum_j p_j$. We emphasize that these probabilities are chosen so that we cancel the effect of misreporting agents in the unconditional probability distribution of $\PPow^{\ell,r}$ and achieve exactly the probability distribution $f^{(\ell,r)}(\vec w_T)$. Moreover, if the mechanism returns $\bot$, we verify all agents. So, it is always possible to compute there probabilities correctly.

The crucial and most technical part of the analysis is to show that $p_j$'s are always non-negative and their sum is at most $1$. To this end, we employ steps~1-4. These steps implement additional verification and ensure that $\Prob[\PPow^{\ell,r}(\vec{x}, \vec{s}) = \bot]$ is large enough for this property to hold (see Section~\ref{s:app:ppower} for the details).

\begin{theorem}\label{th:partial-power}
For every $\eps > 0$, there exist $\ell, r \geq 1$, such that Partial Power is truthful, robust, and $(1 - \eps)$-approximate for the social welfare, and verifies at most $O(\ln m/\eps^2)$ agents in the worst case.
\end{theorem}

%% file: exponential.tex
\section{The Exponential Mechanism with Selective Verification}
\label{s:exponential}

Next, we consider the well known Exponential mechanism and show that it escapes the characterization of Section~\ref{s:impossibility} by relaxing scale invariance.
The Exponential mechanism (or $\Expo$, for brevity) is strongly anonymous and assigns a probability proportional to the exponential of the weight of each outcome.
For each valuation profile $\matr{x}$, the outcome of $\Expo$ depends on $\vec{w} \equiv \sum_{i=1}^n \vec{x}_i$. If all agents are truthful, $\Expo^\alpha(\vec{w})$ results in outcome $j$ with probability $e^{w_j / \alpha} / \sum_{q=1}^m e^{w_q / \alpha}$, i.e., proportional to $e^{w_j / \alpha}$, where $\alpha > 0$ is a parameter.
As in Section~\ref{s:power}, we expand every term $e^{w_j / \alpha}$ and verify only the agents in the tuple $\vec{t}$ corresponding to each term in the expansion below (the sampling can be implemented as in footnote~\ref{foot:sampling}):
\begin{equation}\label{eq:exponential-verification}
 e^{w_j / \alpha} = \sum_{\ell=0}^\infty \frac{(w_j/\alpha)^\ell}{\ell!} = \sum_{\ell=0}^\infty \frac{\alpha^{-\ell}}{\ell!}\sum_{\vec{t} \in N^{\ell}}
 x_{t_1}(j) x_{t_2}(j) \cdots x_{t_\ell}(j)
\end{equation}

The detailed description of $\Expo^\alpha$ is similar to Mechanism~\ref{alg:power}, with the only difference that, in the second step, we pick an outcome $j \in O$, an integer $\ell \geq 0$ and a tuple $\vec{t} \in N^{\ell}$ with probability proportional to the value of the term $x_{t_1}(j) x_{t_2}(j) \cdots x_{t_\ell}(j) / (\alpha^\ell \ell!)$ (see also Mechanism~\ref{alg:exponential} in the Appendix). The following summarizes the properties of $\Expo$.

\begin{theorem}\label{th:exponential}
For any $\alpha > 0$, $\Expo^\alpha(\vec{w})$ is robust and truthful, achieves an additive error of $\alpha\ln m$ wrt. the maximum social welfare and has expected verification $\norm{\vec{w}}_\infty/\alpha$.
\end{theorem}

\begin{proof}[sketch]
Using an argument similar to that used for Power (see Section~\ref{s:app:power}), we can show that $\Expo^\alpha$ is oblivious (note that the allocation of $\Pow^\ell$ is obtained from the allocation of $\Expo^\alpha$ if we condition on a particular exponent $\ell$). Then, robustness follows from Lemma~\ref{l:oblivious}, because $\Expo^\alpha$ is a recursive mechanism. As for participation, the Exponential allocation is known to be MIDR with range $Z = \Delta(O)$ and function $h(\vec{z}) = - \alpha \sum_j z_j \ln z_j$, i.e., $\alpha$ times the entropy of the resulting allocation (see e.g., \cite{HK12}). Therefore, by Lemma~\ref{l:participation}, $\Expo^\alpha$ satisfies participation. Since it is also robust, Lemma~\ref{l:robustness+participation} implies that $\Expo^\alpha$ is truthful.

For the verification, (\ref{eq:exponential-verification}) implies that when all agents are truthful, the number of agents verified, given that the selected outcome is $j$, follows a Poisson distribution with parameter $w_j/\alpha \leq \norm{\vec{w}}_\infty/\alpha$. Therefore, the expected verification is at most $\norm{\vec{w}}_\infty/\alpha$.

As for the approximation guarantee, the optimal social welfare $\norm{\vec{w}}_\infty$ and the objective maximized by $\Expo^\alpha$ differ by $\alpha$ times the entropy of the allocation, which is at most $\alpha \ln m$ (see also Section~\ref{s:app:exponential-approximation}).
\qed\end{proof}

In many settings, we know (or can obtain in a truthful way, e.g., by random sampling) an estimation $E$ of $\norm{\vec{w}}_\infty$ with $E \geq \norm{\vec{w}}_\infty \geq \rho E$, for some $\rho \in (0, 1)$. Then, we can choose $\alpha = \eps \rho E / \ln m$ and obtain an approximation ratio of $1-\eps$ with expected verification $\ln m / (\rho \eps)$, for any $\eps > 0$.
E.g., if for all agents $i$, $|\vec{x}_i| = 1$, $n \geq \norm{\vec{w}}_\infty \geq n/m$. Then, using $\alpha = n \eps / \ln m$, we have an additive error of $\eps n$ with verification $\ln m/\eps$. Moreover, with $\alpha = n \eps / (m \ln m)$, we have approximation ratio $1-\eps$ with verification $m\ln m/\eps$.
Finally, note that, since the number of agents verified follows a Poisson distribution, by Chernoff bounds, the verification bounds hold with high probability in addition to holding in expectation.

%% file: appendix.tex
\newpage
\appendix
\section{Appendix}
\subsection{Other Related Previous Work}
\label{s:app:previous}

The extensive use of monetary transfers in mechanism design is principally because in absence of money, very little can be done to enforce truthfulness. However, there are settings where monetary transfers might be unacceptable, infeasible, or undesirable (see e.g., \cite{PT09} for examples). To circumvent the impossibility result of Gibbard-Satterthwaite in such settings, \cite{PT09} suggested to tradeoff social efficiency for truthfulness and introduced the framework of \emph{approximate mechanism design without money}. The idea is to consider truthful mechanisms without money in a particular domain and determine the best approximation ratio achievable for an appropriate social objective.

In principle, the notion of approximate mechanisms provides the designer with more flexibility. Nevertheless, there have been only few examples of truthful mechanisms with good approximation guarantees that are not based on additional assumptions. All of them concern some simple and restricted domains (see for e.g.,~\cite{AFPT09,FT13,LSWZ10,PT09} for placing $1$ or $2$ facilities in a metric space and \cite{Proc10} for voting with positional scoring rules). For less restricted domains, there are strong lower bounds on the best possible approximation ratio achievable by truthful mechanisms (see e.g., \cite{FT12} for deterministic facility location mechanisms). Therefore, for nontrivial approximation guarantees, we need either some assumptions on the direction or the extent of agent misreports, i.e., to use \emph{verification}, or a way to implicitly penalize misreports, a.k.a. \emph{imposition}.

Probably the most natural and practically applicable notion of verification is \emph{symmetric partial verification} (or \emph{$\eps$-verification}), which explicitly forbids any false declaration at distance larger than $\eps$ to the true type. Interestingly, \cite{AK08,CESY12,FZ13} prove that symmetric partial verification it
does not help in the design of truthful mechanisms (with or without money)! Hence, in order to make some difference in approximate mechanism design without money, verification should be exact, in the sense that it forbids even negligible deviations from the truth, at least for some types of misreports.
Many interesting positive results in approximate mechanism design without money use either ``one-sided'' verification or imposition (see e.g., \cite{FKV13,FT10,Kouts11,NST10,PS14}). However, the use of imposition depends very much on the particular application (see e.g., \cite{FT10,Kouts11,PS14}), while ``one-sided'' verification explicitly forbids a particular type of false declarations for all agents with positive utility (see e.g., \cite{FKV13}). So, through theoretically interesting, ``one-sided'' verification is difficult to apply in practice. Thus, starting from \cite{CESY12}, recent research has focused on the power of exact verification schemes that use either limited or costly verification and mild (or at least bounded) penalties for the liars.

Working in this direction, we seek a better and more delicate understanding of the power of verification in approximate mechanism design without money. Significantly departing from most of the previous work, we develop a general approach to the use of verification in mechanism design without money that is applicable to essentially any domain and does not resort to any explicit (e.g., monetary) penalties that decrease the utility of misreporting agents.

\subsection{Conclusions and Discussion}
\label{s:app:concl}

In this work, we introduce a general approach to approximate mechanism design without money and with selective verification and apply it to the general domain of utilitarian voting (and to Combinatorial Public Project and to $k$-Facility Location). We focus on strongly anonymous randomized mechanisms and characterize such mechanisms that are truthful in expectation, scale invariant, achieve full allocation and have reasonable verification. By relaxing truthfulness, full allocation and scale invariance, we obtain three mechanisms, namely Power, Partial Power, and Exponential, that are truthful (or almost truthful, for Power), achieve an approximation ratio of $1-\eps$ for the social welfare, and verify $O(\ln m/\eps)$ agents, or $O(\ln m/\eps^2)$ agents for Partial Power, where $m$ is the number of outcomes. Hence, we obtained a smoothed tradeoff between the number of agents verified and the quality of approximation. From a technical viewpoint, our mechanisms are based on smooth proportional-like randomized allocation rules. Truthfulness is a consequence of participation, which is closely related to maximal-in-distributional-range, and robustness, which is closely related to obliviousness to the misreporting agents not included in the verification set.

The property of robustness, i.e., namely that the probability distribution of the mechanism does not depend on misreporting agents, seems quite remarkable. To the best of our knowledge, this is the first time that robustness (or a similar) property is considered in mechanism design. Actually, with the possible exception of constant mechanisms, whose probability distribution over outcomes is independent of the agent declarations, a mechanism can be robust only if it uses exact verification.

To see that robustness is a strong property, recall that truthfulness means that a misreporting agent cannot change the allocation \emph{in her favor}, while robustness means that a misreporting agent cannot change the allocation \emph{whatsoever}. Hence, the definition (and the proof) of truthfulness assumes a utility function that each agent maximizes by truthful reporting. Robustness, on the other hand, does not refer to the utility function of the agents. Any misreport that can be caught by the verification oracle does not affect the probability distribution of a robust mechanism, no matter the incentives or the utility function of misreporting agents.

We believe that robustness can be very useful when the agent valuations are not declared explicitly to the mechanism, but they are deduced from their declarations on some observable types. E.g., this happens in the Facility Location domain, where the agents declare their locations to the mechanism, and the definition of truthfulness assumes that each agent wants a facility as close as possible to her declared location and that her disutility increases linearly with the distance (see also \cite{FT13}). On the other hand, robustness only depends on whether each agent declares her true location (e.g., her true home address) to the mechanism, not on whether she wants a facility close, not so close, or far away from her declared location.

\input{applications}

\subsection{The Proof of Lemma~\ref{l:participation}}
\label{s:app:participation}

Let $i$ be any agent. Since the allocation rule $f$ is MIDR, we obtain the following inequalities:
\begin{align*}
 \sum_{j=1}^n \vec{x}_j \cdot f(\matr{x}) + h(f(\matr{x})) & \geq
 \sum_{j=1}^n \vec{x}_j \cdot f(\matr{x}_{-i}) + h(f(\matr{x}_{-i})) \\
 \sum_{j \neq i} \vec{x}_j \cdot f(\matr{x}_{-i}) + h(f(\matr{x}_{-i})) & \geq
 \sum_{j \neq i} \vec{x}_j \cdot f(\matr{x}) + h(f(\matr{x}))
\end{align*}
We apply the MIDR condition to $\matr{x}$, for the first inequality, and to $\matr{x}_{-i}$, for the second inequality. Summing up the two inequalities, we obtain that $\matr{x}_i \cdot f(\matr{x}) \geq \matr{x}_i \cdot f(\matr{x}_{-i})$, i.e., the participation condition. \qed

\subsection{Continuous Allocation Rules with Participation: The Proof of Lemma~\ref{l:chara_weak_parti}}
\label{s:app:chara_weak_parti}

Recall that for any valuation profile $\vec{x}$, the probability distribution of a strongly anonymous allocation rule $f$ depends only on the weight vector $\vec{w}(\vec{x})$ of the outcomes. Hence, we fix a pair of arbitrary weight vectors $\vec{w}, \vec{v} \in \reals_{\geq 0}^m$. For any $a \in (0, 1)$, since $f$ satisfies participation, we have that
\[ (1 - a) \vec{w} \cdot f(a \vec{v} + (1 - a) \vec{w}) \geq
   (1 - a) \vec{w} \cdot f(a \vec{v}) \]
Using that $a \neq 1$ and scale invariance, i.e., that $f(a \vec{v}) = f(\vec{v})$, we obtain that
\[ \vec{w} \cdot f(a \vec{v} + (1 - a) \vec{w}) \geq \vec{w} \cdot f(\vec{v}) \]
This holds for every $a \in (0, 1)$. Thus, taking the limit as $a$ goes to $0$ and using the hypothesis that $f$ is continuous, we get that for all weight vectors $\vec{w}, \vec{v} \in \reals^m_{\geq 0}$,
\[ \vec{w} \cdot f(\vec{w}) \geq \vec{w} \cdot f(\vec{v}) \]

We can now define $Z$ as the image set of $f$, i.e.,
$Z = \{ \vec{z} ~|~ \exists \vec{v}\mbox{ such that }\vec{z} = f(\vec{v}) \}$.
Using the previous inequality, we obtain that for all weight vectors $\vec{w}$ and all $\vec{z} \in Z$,
$\vec{w} \cdot f(\vec{w}) \geq \vec{w} \cdot \vec{z}$.
Since $f(\vec{w}) \in Z$, a necessary condition for $f$ to satisfy participation is that $f(\vec{w}) = \arg \max_{z \in Z} \vec{w} \cdot \vec{z}$.
\qed

\subsection{The Proof of Lemma~\ref{l:robustness+participation}}
\label{s:app:robustness+participation}

Since $F$ is robust, for any agent $i$, for any valuation pair $\vec{x}_i$ and $\vec{y}_i$ and for all reported valuations $\matr{y}_{-i}$ and verification vectors $\vec{s}_{-i}$,
\begin{align*}
 F((\matr{y}_{-i}, \vec{x}_i), (\vec{s}_{-i}, 1)) & =
 F(((\matr{y}_T)_{-i}, \vec{x}_i), (1, \ldots ,1)) \mbox{\ \ \ and}\\
 F((\matr{y}_{-i}, \vec{y}_i), (\vec{s}_{-i}, 0)) & = F((\matr{y}_T)_{-i}, (1, \ldots ,1))
\end{align*}
Here, we assume that $\vec{x}_i$ is $i$'s true type and $\vec{y}_i \neq \vec{x}_i$ is a misreport.
Moreover, using that $f$ (i.e., the allocation rule induced by $F$ on truthful reports) satisfies $\e$-participation, we have that:
\[ \vec{x}_i \cdot F(((\matr{y}_T)_{-i}, \vec{x}_i), (1, \ldots, 1)) \geq
   \e\,\vec{x}_i \cdot F((\matr{y}_T)_{-i}, (1, \ldots, 1)) \]
Combining the three equations above, we obtain that
\[ \vec{x}_i \cdot F((\matr{y}_{-i}, \vec{x}_i), (\vec{s}_{-i}, 1)) \geq
   \e\,\vec{x}_i \cdot F((\matr{y}_{-i}, \vec{y}_i), (\vec{s}_{-i}, 0)\,, \]
i.e., that the mechanism with verification $F$ is $\e$-truthful.
\qed

%% file: applications.tex
\subsection{An Application to Combinatorial Public Project}
\label{s:app:applications}

The \emph{Combinatorial Public Project Problem} (CPPP) was introduced in \cite{SS08,PSS08} and is a well-studied problem in algorithmic mechanism design. An instance of CPPP consists of a set $R$ with $r$ resources, a parameter $k$, $1 \leq k \leq r$, and $n$ strategic agents, where each agent $i$ has a function $\vec{x}_i : 2^R \to \reals_{\geq 0}$ that assigns a non-negative valuation $\vec{x}_i(S)$ to each resource subset $S \subseteq R$. The objective is to find a set $C$ of $k$ resources that maximizes $\sum_{i} \vec{x}_i(C)$, i.e., the social welfare of the agents from $C$. We assume that all valuations $\vec{x}_i$ are \emph{normalized}, i.e., $\vec{x}_i(\emptyset) = 0$, and \emph{monotone}, i.e., $\vec{x}_i(S_1) \leq \vec{x}_i(S_2)$ for all $S_1 \subseteq S_2$.

The valuation functions $\vec{x}_i$ are implicitly represented through a value oracle, which returns the valuation $\vec{x}_i(S)$ of any resource subset $S$ in $O(1)$ time. Then, CPPP is $\NP$-hard and practically inapproximable in polynomial time, under standard computational complexity assumptions (see \cite{SS08} for the details). If the valuation functions $\vec{x}_i$ are \emph{submodular}, i.e., each $\vec{x}_i$ satisfies $\vec{x}_i(S_1 \union S_2) + \vec{x}_i(S_1 \cut S_2) \leq \vec{x}_i(S_1)+\vec{x}_i(S_2)$, for all $S_1, S_2 \subseteq R$, CPPP can be approximated in polynomial time within a factor of $1-1/e$. If the valuation functions $\vec{x}_i$ are \emph{subadditive}, i.e., each $\vec{x}_i$ satisfies $\vec{x}_i(S_1 \union S_2)\leq \vec{x}_i(S_1)+\vec{x}_i(S_2)$, for all $S_1, S_2 \subseteq R$, CPPP can be approximated in polynomial time within a factor of $r^{-1/2}$, while approximating it within any factor better than $r^{-1/4+\eps}$, for any constant $\eps > 0$, requires exponential communication \cite{SS08}.

In \cite{PSS08}, it was shown that for submodular valuations, CPPP cannot be approximated in polynomial time (or with polynomial communication) by deterministic truthful mechanisms (with money) within any factor better than $r^{-1/2+\eps}$, for any constant $\eps > 0$. A similar communication complexity lower bound was shown in \cite{Dob11} for randomized truthful in expectation mechanisms with money. So, the polynomial-time approximability of CPPP with submodular valuations is dramatically  better than its approximability by polynomial-time truthful mechanisms with money. Although the approximability of CPPP by polynomial time truthful mechanisms with money has received considerable attention, to the best of our knowledge, this is the first time that the approximability of CPPP by truthful mechanisms without money is considered. 

CPPP, with general valuation functions, can be naturally cast in our framework of utilitarian voting, with the outcome set $O$ consisting of all resource subsets $S$ with $|S| = k$ (hence, we have $m \leq r^k$). Then, our mechanisms imply the following results on the approximability of CPPP with general valuation functions by mechanisms without money and with selective verification:

\begin{description}
\item[Power.] For any $\eps > 0$, the Power mechanism always allocates a set of $k$ resources, is robust, $\eps$-truthful, achieves an approximation ratio of $1 - \eps$ and verifies at most $k \log r / \eps$ agents.
    
\item[Partial Power.] For any $\eps > 0$, the Partial Power mechanism allocates a set of $k$ resources with probability $1-O(\eps)$, is robust, truthful, achieves an approximation ratio of $1 - \eps$ and verifies $O(k \log r / \eps^2)$ agents. Note that the empty set can naturally play the role of the null outcome for Partial Power. 
    
\item[Exponential.] Since Exponential is not scale invariant, we need to assume that $\max_{S \subseteq R, |S| \leq k} \vec{x}_i(S) \leq 1$, for every agent $i$. Then, for any $\eps > 0$, the Exponential mechanism always allocates a set of $k$ resources, is robust, truthful, and achieves an additive error of $\eps n$ with verification of $O(k \log r / \eps)$ agents, or achieves an approximation ratio of $1 - \eps$ with verification of $O(k r^k \log r / \eps)$ agents (where the verification bounds hold with high probability).
\end{description}

These guarantees are very strong and rather surprising, especially if the number $n$ of agents is significantly larger than $k \log r$, which is the case in many practical settings. We almost reach the optimal social welfare of the famous VCG mechanism, which achieves truthfulness through (potentially very large) payments, using truthful mechanisms without money that verify a small number of agents independent of $n$. It becomes even more interesting if we recall that the penalty for a misreporting agent, through which we enforce truthfulness, is just the exclusion of the agent's preferences from the decision making process.

The mechanisms above run in time polynomial in the total number of outcomes $r^k$ and in the number of agents $n$. So, if the valuation functions are implicitly represented by value oracles, they are not computationally efficient. However, we still need to resort to approximate solutions, because, in absence of money, the optimal solution is not truthful. We underline that computational inefficiency is unavoidable, since our approximation ratio of $1-\eps$, for any constant $\eps > 0$, is dramatically better than known impossibility results on the polynomial time approximability of CPPP.

If we insist on computationally efficient mechanisms without money for CPPP, we can combine our mechanisms with existing maximal-in-range mechanisms so that everything runs in polynomial time. E.g., for CPPP with subadditive valuation functions, we can use the maximal-in-range mechanism of \cite[Sec.~3.2]{SS08} and obtain randomized polynomial-time truthful mechanisms without money that achieve an approximation ratio of $O(\min\{k, \sqrt{r}\})$ for the social welfare with selective verification of $O(k \log r)$ agents.

%% file: appendix_facility.tex
\subsection{$k$-Facility Location: Minimizing the Maximum Cost}
\label{s:app:k-center}

We first introduce some pieces of notation and terminology that we need for this section and for Section~\ref{s:app:k-facility-location}.
For $k$-Facility Location problems, we consider an underlying metric space $(M, d)$, where $M$ is a finite set of points and $d : M \times M \mapsto \reals_{\geq 0}$ is a distance function, which is non-negative, symmetric, and satisfies the triangle inequality.
For each $t \in M$ and $M' \subseteq M$, we let $d(t, M') = \min\{ d(t, t') : t' \in M' \}$.

The outcomes are all subsets of $k$ locations in $M$, i.e., all $C \subseteq M$ with $|C| = k$. Each agent $i$ has a preferred location $t_i \in M$ and her ``valuation'' for outcome $C$ is $\vec{x}_i(C) = -d(t_i, C)$, i.e., minus the distance of her preferred location to the nearest facility in $C$. The minus sign is due to the cost minimization nature of the problem. So, each agent $i$ aims at minimizing $d(t_i, C)$.

A (possibly randomized) mechanism $F$ takes as input a profile $\vec{z} = (z_1, \ldots, z_n)$ of reported locations%
\footnote{Note that in Facility Location, each agent declares an ``observable'' type, namely, her preferred location, to the mechanism and not her entire valuation / cost
function, as she does in the utilitarian voting domain. The mechanism deduces her cost function based on her reported ``observable'' type.}.
Using oracle access to a verification vector $\vec{s}$, the mechanism $F$ maps $\vec{z}$ to a set $C \subseteq M$ of $k$ facilities. The cost of each agent $i$ is the expected distance of her true location $t_i$ to the nearest facility in $C$. For clarity, we denote $i$'s expected cost as
\( \cost[t_i, F(\vec{z}, \vec{s})] = \Exp_{C \sim F(\vec{z}, \vec{s})}[d(t_i, C)] \), %
which agent $i$ seeks to minimize.

The main properties (e.g., truthfulness, robustness, approximation ratio) of a Facility Location mechanism with verification are defined as in Sections~\ref{s:prelim}~and~\ref{s:model}, with the only difference that the agents now seek to minimize their costs.
E.g., a mechanism with verification $F$ is truthful if for any agent $i$, for any location pair $t_i$ and $z_i$, and for all reported locations $\vec{z}_{-i}$ and verification vectors $\vec{s}_{-i}$,
\[  \cost[t_i, F((\vec{z}_{-i}, t_i), (\vec{s}_{-i}, 1))] \leq
    \cost[t_i, F((\vec{z}_{-i}, z_i), (\vec{s}_{-i}, 0))] \]
%
%Similarly, $F$ satisfies the participation constraint if for any agent $i$ with true location $t_i$ and for all reported locations $\vec{z}_{-i}$ and verification vectors $\vec{s}_{-i}$,
%
%\(  \cost[t_i, F((\vec{z}_{-i}, t_i), (\vec{s}_{-i}, 1))] \leq
%    \cost[t_i, F(\vec{z}_{-i}, \vec{s}_{-i})] \).

In this section, we focus on minimizing $\max_{i \in N} \{ d(t_i, F(\vec{t}, \vec{1})) \}$, i.e., the maximum distance of any agent to the nearest facility. The Greedy mechanism (or $G$, for brevity, see  Mechanism~\ref{alg:k-center}) is a truthful and robust extension of the $2$-approximate Greedy algorithm for $k$-Center (see e.g., \cite[Section~2.2]{WS10}). In the description of Greedy, we write $i \in \vec{z}$ to denote that the reported location of agent $i$ participates in the location profile $\vec{z}$. Also, we assume that ties are broken in some fixed deterministic way. We next prove Theorem~\ref{th:k-center}, stated in Section~\ref{s:facility}.

\begin{algorithm}[t]
\caption{\label{alg:k-center}The Greedy Mechanism $G(\vec{z}, \vec{s})$ for $k$-Facility Location}
\begin{algorithmic}\normalsize
    \State $L \leftarrow \emptyset$\,;
    \ \ \ Pick the first agent $i \in \vec{z}$ and let $C \leftarrow \{ z_i \}$
    \While{$|C| < k$}
        \State $i \leftarrow \arg\max_{i \in \vec{z}} d(z_i, C)$
        \State $C \leftarrow C \union \{ z_i \}$
    \EndWhile
    \For{all $z_i \in C$}
        \If{$\ver(i) \neq 1$} $L \leftarrow L \union \{ i \}$\EndIf
    \EndFor
    \If{$L \neq \emptyset$} \Return $G(\vec{z}_{-L}, \vec{s}_{-L})$
    \Else\ \Return $C$
    \EndIf
\end{algorithmic}\end{algorithm}

\begin{proof}[of Theorem~\ref{th:k-center}]
Clearly, $\Ver(G) = k$, since if all agents are truthful, Greedy verifies only the agents allocated a facility. Moreover, the approximation ratio of Greedy is $2$ (see e.g., \cite[Theorem~2.3]{WS10}).

By Lemma~\ref{l:robustness+participation}, to show that Greedy is truthful, it suffices to show that Greedy is robust and satisfies participation. As for the latter, we fix an agent $i$ with true location $t_i$ and a location profile $\vec{z}_{-i}$. If $t_i$ is allocated a facility in $(\vec{z}_{-i}, t_i)$, $i$'s cost is $0$. Otherwise, excluding $t_i$ from $(\vec{z}_{-i}, t_i)$ does not affect the execution of Greedy and $i$'s cost is the same in $(\vec{z}_{-i}, t_i)$ and in $\vec{z}_{-i}$.

By Lemma~\ref{l:oblivious}, robustness follows from the obliviousness of Greedy, since Greedy is a recursive mechanism with verification. As for obliviousness, let us assume that all agents in $C$ are truthful. Then, for any set $L$, with $C \cut L = \emptyset$, the outcome of Greedy on both $\vec{z}$ and $\vec{z}_{-L}$ is $C$. In words, if all agents in $C$ are truthful, the agents not in $C$ do not affect the outcome of Greedy. Therefore, Greedy is oblivious and robust.
\qed\end{proof}

\subsection{$k$-Facility Location: Minimizing the Social Cost}
\label{s:app:k-facility-location}

In this section, we focus on minimizing $\sum_{i=1}^n d(t_i, F(\vec{t}, \vec{1}))$, i.e, the total distance of the agents to the nearest facility of the mechanism. One may regard this version of $k$-Facility Location as a welfare maximization problem but with negative valuations.

Mechanism~\ref{alg:k-facility-location} (or $P$, for brevity) is a truthful and robust extension of the Proportional mechanism \cite{LSWZ10}, which achieves an approximation ratio of $\Theta(\ln k)$ for the objective of social cost \cite{AV07}. Mechanism~\ref{alg:k-facility-location} is essentially a randomized version of Mechanism~\ref{alg:k-center}. We next prove Theorem~\ref{th:k-facility-location}, stated in Section~\ref{s:facility}.

\begin{algorithm}[t]
\caption{\label{alg:k-facility-location}The Proportional Mechanism $P(\vec{z}, \vec{s})$ for $k$-Facility Location}
\begin{algorithmic}\normalsize
\State $L \leftarrow \emptyset$\,;
    \ \ \ Pick an agent $i \in \vec{z}$ uniformly at random
    and let $C \leftarrow \{ z_i \}$
    \While{$|C| < k$}
        \State pick an agent $i \in \vec{z}$ with probability
               $d(z_i, C)/\sum_{j \in \vec{z}} d(z_j, C)$
        \State $C \leftarrow C \union \{ z_i \}$
    \EndWhile
    \For{all $z_i \in C$}
        \If{$\ver(i) \neq 1$} $L \leftarrow L \union \{ i \}$\EndIf
    \EndFor
    \If{$L \neq \emptyset$} \Return $P(\vec{z}_{-L}, \vec{s}_{-L})$
    \Else\ \Return $C$
    \EndIf
\end{algorithmic}\end{algorithm}

\begin{proof}[of Theorem~\ref{th:k-facility-location}]
The approximation ratio of Mechanism~\ref{alg:k-facility-location} is shown in \cite[Theorem~5.1]{AV07}. Moreover, $\Ver(P) = k$, since if all agents report truthfully, Proportional verifies only the agents allocated a facility. By Lemma~\ref{l:robustness+participation}, it suffices to show that Proportional is robust and satisfies participation. The proof is a generalization of the proof of Theorem~\ref{th:k-center}.

By Lemma~\ref{l:oblivious}, robustness follows from the obliviousness of Proportional, since Proportional is a recursive mechanism with verification. The obliviousness of Proportional was first observed in \cite{FT10}. Since our setting is different, we include a proof here for completeness. So, we next show that for all location profiles $\vec{z}$ and verification vectors $\vec{s}$, with $L = N \setminus T(\vec{s})$, and all outcomes $C$,
\begin{equation}\label{eq:prop_oblivious}
  \Prob[ P(\vec{z}, \vec{s}) = C\,|\,C \cut L = \emptyset] =
  \Prob[ P(\vec{z}_{-L}, \vec{1}) = C]
\end{equation}
which implies that Proportional is oblivious (note also that any possible $C \not\subseteq T(\vec{s})$ has probability $0$).
To establish (\ref{eq:prop_oblivious}), we observe that conditional on the event that no misreporting agent is selected in $C$ by $P(\vec{z}, \vec{s})$, the probability distribution of $P(\vec{z}, \vec{s})$ is identical to the probability distribution of $P(\vec{z}_{-L}, \vec{1})$.
This claim is shown by induction on the number of selected agents and a simple coupling argument. To this end, let us fix a location profile $\vec{z}$ and the set $L$ of misreporting agents in $\vec{z}$. The first agent is selected uniformly at random from $N \setminus L$ both by $P(\vec{z}, \vec{s})$, conditional on the event that the selected agent is not in $L$, and by $P(\vec{z}_{-L}, \vec{1})$.
Assume inductively that both executions $P(\vec{z}, \vec{s})$, conditional on $C \cut L = \emptyset$, and $P(\vec{z}_{-L}, \vec{1})$ agree on the selected set $C$ up to some point. Then, the next agent in $C$ is selected by both executions from exactly the same probability distribution, since due to the conditioning, $P(\vec{z}, \vec{s})$ does not consider any agents in $L$. Finally, since the set $C$ of $k$ agents selected by both executions does not include any agents from $L$, both executions result in $C$ with identical probability.

We next show that Mechanism~\ref{alg:k-facility-location} satisfies participation. We fix an agent $i$ with location $t_i$ and a location profile $\vec{t}_{-i}$.
Since participation is a property of the allocation rule, we assume that all agents are truthful and completely ignore verification from this point on. So, we simply write $P(\vec{t})$, instead of $P(\vec{t}, \vec{s})$.

For each round $\ell = 0, 1, \ldots, k$, we let $\cost[t_i, P(\vec{t})| C_\ell]$ be the expected cost of $i$ at the end of the $P$, given that the facility set of $P(\vec{t})$ at the end of round $\ell$ is $C_\ell$.
Similarly, we let $\cost[t_i, P(\vec{t}_{-i})| C_\ell]$ be the expected cost of $i$ at the end of the $P(\vec{t}_{-i})$, given that $i$ does not participate in the mechanism and that the facility set of $P$ at the end of round $\ell$ is $C_\ell$.

For $\ell = k$, $\cost[t_i, P(\vec{t})| C_k] = \cost[t_i, P(\vec{t}_{-i})| C_k] = d(t_i, C_k)$. For each round $\ell = 1, \ldots, k-1$, if $i$ participates in the mechanism, with probability proportional to $d(t_i, C_\ell)$ the next facility is placed at $t_i$, in which case $i$'s cost is $0$, while for each agent $j \neq i$, with probability proportional to $d(t_j, C_\ell)$ the next facility is placed at $t_j$, in which case the expected cost of $i$ is $\cost[t_i, P(\vec{t})| C_\ell \union \{ t_j \}]$. Therefore, the expected cost of $i$ is:
\begin{equation}\label{eq:rec-cost}
 \cost[t_i, P(\vec{t}) | C_\ell] =
 \frac{\sum_{j \neq i} d(t_j, C_\ell)\,\cost[t_i, P(\vec{t})|C_\ell \union \{t_j\}]}
 { d(t_i, C_\ell) + \sum_{j \neq i} d(t_j, C_\ell)}
\end{equation}
For $\ell = 0$, the expected cost of agent $i$ is:
\begin{equation}\label{eq:rec-cost0}
 \cost[t_i, P(\vec{t})] = \frac{\sum_{j \neq i} \cost[t_i, P(\vec{t})|\{ t_j \}]}{n}
\end{equation}

If $i$ does not participate in the mechanism, her expected cost for $\ell = 1, \ldots, k$ is:
\begin{equation}\label{eq:rec-cost-nonp}
 \cost[t_i, P(\vec{t}_{-i}) | C_\ell] =
 \frac{\sum_{j \neq i} d(t_j, C_\ell)\,\cost[t_i, P(\vec{t}_{-i})|C_\ell \union \{t_j\}]}
 { \sum_{j \neq i} d(t_j, C_\ell)}
\end{equation}
If $i$ does not participate in the mechanism, her expected cost for $\ell = 0$ is:
\begin{equation}\label{eq:rec-cost0-nonp}
 \cost[t_i, P(\vec{t}_{-i})] = \frac{\sum_{j \neq i} \cost[t_i, P(\vec{t}_{-i})|\{ t_j \}]}{n-1}
\end{equation}

Using induction on $\ell$, we next show that for any round $\ell = 0, 1, \ldots, k$ and any set $C_\ell$,
\begin{equation}\label{eq:prop-ind}
 \cost[t_i, P(\vec{t})|C_\ell] \leq \cost[t_i, P(\vec{t}_{-i})|C_\ell]
\end{equation}
Clearly, (\ref{eq:prop-ind}) implies that $P$ (Mechanism~\ref{alg:k-facility-location}) satisfies participation.

For the basis of the induction, we observe that (\ref{eq:prop-ind}) holds trivially for $\ell = k$. We inductively assume that (\ref{eq:prop-ind}) holds for $\ell+1$ and any facility set $C_{\ell+1}$, and show that (\ref{eq:prop-ind}) holds for $\ell$ and any facility set $C_\ell$.
If $\ell \geq 1$, we use (\ref{eq:rec-cost-nonp}) and obtain that:
\begin{eqnarray*}
 \cost[t_i, P(\vec{t}_{-i})|C_\ell] & = &
  \frac{\sum_{j \neq i} d(t_j, C_\ell)\,\cost[t_i, P(\vec{t}_{-i})|C_\ell \union \{t_j\}]}
       {\sum_{j \neq i} d(t_j, C_\ell)} \\
 & \geq & \frac{\sum_{j \neq i} d(t_j, C_\ell)\,\cost[t_i, P(\vec{t})|C_\ell \union \{t_j\}]}
               {d(t_i, C_\ell) + \sum_{j \neq i} d(t_j, C_\ell)}
 = \cost[t_i, P(\vec{t})|C_\ell]
\end{eqnarray*}
The inequality follows from the induction hypothesis and that $d(t_i, C_\ell) \geq 0$. The second equality is (\ref{eq:rec-cost}).

If $\ell = 0$, we use (\ref{eq:rec-cost0-nonp}) and obtain that:
\begin{eqnarray*}
 \cost[t_i, P(\vec{t}_{-i})] & = &
 \frac{\sum_{j \neq i} \cost[t_i, P(\vec{t}_{-i})|\{ t_j \}]}{n-1} \\
 & \geq & \frac{\sum_{j \neq i} \cost[t_i, P(\vec{t})|\{ t_j \}]}{n}
 = \cost[t_i, P(\vec{t})]
\end{eqnarray*}
The inequality follows from the induction hypothesis and the fact that $n > n-1$. The second equality is (\ref{eq:rec-cost0}). Thus we have established (\ref{eq:prop-ind}) for any round $\ell = 0, 1, \ldots, k$, and any facility set $C_\ell$.

To conclude the proof of Theorem~\ref{th:k-facility-location}, we observe that since $P$ is robust and satisfies participation, it is also truthful, by Lemma~\ref{l:robustness+participation}.
\qed\end{proof}

%% file: appendix_power.tex
\subsection{The Proof of Theorem~\ref{th:power}}
\label{s:app:power}

We establish the properties of $\Pow^\ell$ for any $\ell \geq 0$. To obtain Theorem~\ref{th:power}, we set $\ell = \ln m / \eps$.

If all agents are truthful, $\Pow^\ell$ picks an outcome $j$ and a tuple $\vec{t} \in N^{\ell}$ and returns $j$ after verifying all agents in $\vec{t}$. Since there are at most $\ell$ different agent indices in $\vec{t}$, the verification of $\Pow^\ell$ is at most $\ell$.

Due to its proportional nature, $\Pow^\ell$ is oblivious to the declarations of misreporting agents not verified (see also the proof of obliviousness for Proportional, in Section~\ref{s:app:k-facility-location}). Specifically, we show that for all valuation profiles $\matr{x}$ and verification vectors $\vec{s}$, with $L = N \setminus T(\vec{s})$, and all outcomes $j$,
\begin{equation}\label{eq:power_oblivious}
\Prob[\Pow^\ell(\matr{x}, \vec{s}) = j\,|\, V(\matr{x}) \cut L = \emptyset ] = \Prob[\Pow^\ell(\matr{x}_{-L}, \vec{1}) = j]
\end{equation}
Therefore, $\Pow^\ell$ is oblivious. Since it is also recursive, Lemma~\ref{l:oblivious} implies that $\Pow^\ell$ is robust.

To prove (\ref{eq:power_oblivious}), we observe that for any outcome $j$, the condition $V(\matr{x}) \cut L = \emptyset$ implies that $\Pow^\ell(\matr{x}, \vec{s})$ selects only terms $x_{t_1}(j) x_{t_2}(j) \cdots x_{t_\ell}(j)$ with truthful agents in $T(\vec{s})$. Every term with some valuation $x_{t'}(j)$ of a misreporting agent $t' \in L$ is excluded, since $V(\matr{x}) \cut L = \emptyset$ implies that $\vec{t} \in T(\vec{s})^\ell$. Therefore, for any outcome $j$, $\Pow^\ell(\matr{x}, \vec{s})$, conditional on $V(\matr{x}) \cut L = \emptyset$, and $\Pow^\ell(\matr{x}_{-L}, \vec{1})$ have exactly the same set of ``allowable'' terms from which they select $x_{t_1}(j) x_{t_2}(j) \cdots x_{t_\ell}(j)$ and $\vec{t}$. In both, each such term is selected with probability proportional to its value. %, i.e., with the same probability.
Taking all outcomes into account, we obtain that the distribution of $\Pow^\ell(\matr{x}, \vec{s})$, conditional on $V(\matr{x}) \cut L = \emptyset$, and the distribution of $\Pow^\ell(\matr{x}_{-L}, \vec{1})$ are identical.

We next establish the approximation ratio of $\Pow^\ell$ for the objective of maximizing the social welfare. The intuition is that as $\ell$ increases from $0$ to $\infty$, the probability distribution of $\Pow^\ell$ sharpens from the uniform allocation, where each outcome is selected with probability $1/m$, to the optimal allocation. The rate of this transition determines the approximation ratio and is quantified by the following.

\begin{lemma}\label{l:power-approx}
For any $\ell \geq 0$, $\Pow^\ell$ is $m^{-1/(\ell+1)}$-approximate for the social welfare.
\end{lemma}

\begin{proof}
Let us fix any valuation profile $\matr{x}$ and let $\vec{w} \equiv  \vec{w}(\matr{x})$ be the outcome weights in $\matr{x}$. For the approximation ratio, we can assume that all agents are truthful. So, we let $\Pow^\ell(\vec{w}) \equiv \Pow^\ell(\matr{x}, \vec{1})$, for convenience.

The optimal social welfare is $\|\vec{w}\|_{\infty}$. The expected social welfare of the mechanism is
$\vec{w} \cdot \Pow^\ell(\vec{w}) = |\vec{w}^{\ell + 1}|/|\vec{w}^{\ell}|$.
So the approximation ratio of $\Pow^\ell$ is equal to:
\[ \frac{|\vec{w}^{\ell + 1}|}{|\vec{w}^{\ell }|\,\|\vec{w}\|_{\infty}} =
   \frac{(\|\vec{w}\|_{\ell+1})^{\ell+1}}
        {(\|\vec{w}\|_{\ell})^{\ell}\,\|\vec{w}\|_{\infty}} =
   \left( \frac{\|\vec{w}\|_{\ell+1}}{\|\vec{w}\|_{\ell}} \right)^{\!\!\ell}
   \frac{\|\vec{w}\|_{\ell+1}}{\|\vec{w}\|_{\infty}}
\]
Using that $\|\vec{w}\|_{\infty} \leq \|\vec{w}\|_{\ell+1}$ and that
$\|\vec{w}\|_{\ell} \leq m^{\left( \frac{1}{\ell} - \frac{1}{\ell+1} \right)} \|\vec{w}\|_{\ell+1} = m^{1/\ell(\ell+1)} \|\vec{w}\|_{\ell+1}$,
we obtain that the approximation ratio of $\Pow^\ell$ is at least $m^{-1/(\ell+1)}$.
\qed\end{proof}

Unfortunately, Power does not satisfy participation. For a simple example, let $m = 2$ and $n = 2$ and let $\vec{x}_1 = (1, 0)$ and $\vec{x}_2 = (3/4, 1/4)$. Then, agent $2$ prefers outcome $1$, but her participation decreases its probability from $1$, when agent $1$ is alone, to something less than $1$, when both agents participate.

However, Power satisfies participation approximately. To prove this, we use the fact that the Partial Power allocation (see Section~\ref{s:partial}) is MIDR, by definition, and essentially a slightly ``curved'' version of Power. Using that the probabilities that each outcome is selected in Partial Power and in Power are close to each other and the fact that Partial Power satisfies participation, we obtain the following. Since $\Pow^\ell$ is robust and satisfies $m^{-1/(\ell+1)}$-participation, Lemma~\ref{l:robustness+participation} implies that $\Pow^\ell$  is $m^{-1/(\ell+1)}$-truthful.

\begin{lemma}\label{l:power-participation}
For any $\ell \geq 0$, $\Pow^\ell$ satisfies $m^{-1/(\ell+1)}$-participation.
\end{lemma}

\begin{proof}
Let us fix any valuation profile $\matr{x}$ and let $\vec{w}\equiv \vec{w}(\matr{x})$ be the outcome weights in $\matr{x}$. Since participation is a property of the allocation rule, we can assume wlog. that all agents are truthful in $\matr{x}$. So, we let $\Pow^\ell(\vec{w}) \equiv \Pow^\ell(\matr{x}, \vec{1})$, for convenience.

We need to show that for any agent $i$, her expected utility wrt. the probability distribution of $\Pow^{\ell}(\vec{w})$ is no less than $m^{-1/(\ell+1)}$ times her expected utility wrt. the probability distribution of $\Pow^{\ell}(\vec{w}_{-i})$. Since $\Pow^{\ell}(\vec{w})$ selects each outcome $j$ with probability $w^\ell_j / |\vec{w}^{\ell }|$, we need to show that for any agent $i$,
\[
   m^{\frac{1}{\ell+1}}\, \frac{\vec{x}_i \cdot \vec{w}^{\ell }}{|\vec{w}^{\ell }|}
 \geq \frac{\vec{x}_i \cdot \vec{w}^{\ell }_{-i}}{|\vec{w}^{\ell }_{-i}|}
\]

In Section~\ref{s:app:partial}, Lemma~\ref{l:optimization}, we consider the Partial Power allocation with parameter $\ell$ (and with $1-1/r = 1$) and show that it is MIDR with range $Z_\ell = \left\{ \vec{z} \in \reals^m_{\geq 0} : \| \vec{z} \|_{1+1/\ell} \leq m^{-1/(\ell+1)} \right\}$ and function $h(\vec{z}) = 0$. Therefore, by  Lemma~\ref{l:participation}, Partial Power satisfies participation. To establish that $\Pow^\ell$ satisfies participation, we observe that Partial Power with parameter $\ell$ is essentially a slightly ``curved'' version of $\Pow^\ell$, where each outcome $j$ is selected with probability $\frac{w^\ell_j}{m^{1/(1+\ell)} \|\vec{w}^{\ell }\|_{1+1/\ell}}$, instead of $w^\ell_j / |\vec{w}^{\ell }|$ in $\Pow^\ell$. Using the fact that Partial Power satisfies participation, we obtain that for any agent $i$,
\[ \frac{\vec{x}_i \cdot \vec{w}^{\ell }}{\|\vec{w}^{\ell }\|_{1+\frac{1}{\ell }}} \geq  \frac{\vec{x}_i \cdot \vec{w}^{\ell}_{-i}}{\|\vec{w}_{-i}^{\ell}\|_{1+\frac{1}{\ell}}} \]

Since for any $m$-dimensional vector $\vec y$, $|\vec{y}| \ge \|\vec{y}\|_{1+1/\ell} \geq \frac{1}{m^{1 / (\ell+1)}}\,|\vec{y}|$, we have that
\[  m^{\frac{1}{\ell+1}}\, \frac{\vec{x}_i \cdot \vec{w}^{\ell }}{|\vec{w}^{\ell }|} \ge \frac{\vec{x}_i \cdot \vec{w}^{\ell }}{\|\vec{w}^{\ell }\|_{1+\frac{1}{\ell }}}
\ge \frac{\vec{x}_i \cdot \vec{w}^{\ell }_{-i}}{\|\vec{w}_{-i}^{\ell}\|_{1+\frac{1}{\ell}}} \ge   \frac{\vec{x}_i \cdot \vec{w}^{\ell }_{-i}}{|\vec{w}^{\ell }_{-i}|}  \]
Thus, $\Pow^\ell$ satisfies $m^{-1/(\ell+1)}$-participation.
\qed\end{proof}

\subsection{A Detailed Proof of Theorem~\ref{th:lower_bound}}
\label{s:app:lower_bound}

In this section, we give a detailed proof of Theorem~\ref{th:lower_bound} from scratch, for completeness and clarity. Recall that the proof describes a random family of instances where any strongly anonymous truthful mechanism needs to verify a logarithmic number of agents in expectation.

To define such instances, we consider $m$ outcomes and $m$ disjoint groups of agents. Each group has a large number $\nu$ of agents. An agent in group $j$ has valuation $0$ for any outcome other than $j$ and valuation either $1$ or $\delta$ for outcome $j$, where $\delta > 0$ is extremely small (e.g., $\delta = 1/\nu^{10}$). In each group $j$, the probability that exactly $k$, $0 \leq k \leq \nu$, agents have valuation $1$ for outcome $j$ is $2^{-(k+1)}$.

The expected total valuation of each outcome is essentially $1$, since $\nu$ is very large and $\nu\delta$ is negligible. The probability that a group $j$ has at least $\log m$ agents with valuation $1$ is $2^{-\log m} = 1/m$. Hence, the probability that some group has at least $\log m$ agents with valuation $1$ is $1-(1-1/m)^m \geq 1-\exp(-1)$. A similar analysis shows that the probability that some group has more than $2\log m$ agents with valuation $1$ is at most $1/m$. Therefore, the expected maximum social welfare of such instances is $\Theta(\log m)$. In fact, one can show that a random instance has maximum social welfare $\Theta(\log m)$ with high probability.

We next consider any truthful mechanism $F$ with expected social welfare of $\Omega(\log m)$ for such instances and show that $F$ needs to verify a logarithmic number of agents. To this end, we focus on a group $j$ of agents and fix $\vec{x}_{-j}$, i.e., the declarations of the agents in all other groups.  Now we can assume wlog. that the probability of outcome $j$ depends only on the number of agents in group $j$ that declare valuation $1$ for outcome $j$. To see that this is indeed wlog., we first observe that the number of agents that declare $1$ for $j$ fully determines the number of agents that declare $\delta$ for $j$. Moreover, if the probability of outcome $j$ depends on the identities of the agents that declare $1$, then we could randomly permute the agents in group $j$ (and in all other groups) before running the mechanism $F$. This would lead to an anonymous mechanism with the same approximation and verification guarantees. Thus wlog. for any fixed $\vec{x}_{-j}$, $F$ should induce a sequence of probabilities $p_0, p_1, \ldots, p_k, \cdots$, where $p_k$ is the probability of outcome $j$, given that the number of agents that declare valuation $1$ for outcome $j$ is $k$.
Since the mechanism $F$ is truthful, if $k$ agents declare valuation $1$ for outcome $j$, we need to verify each of them with probability at least $p_k - p_{k-1}$. To see this, observe that the expected utility of an agent with valuation $\delta$ for outcome $j$ is $p_{k-1}\delta$, if she is truthful, and at least $p_{k}\delta$ times the probability that she is not verified, if she is not truthful and declares $1$. Since the former is no less than the later (and since $p_k \in [0, 1]$), $F$ needs to verify any agent that declares $1$ for $j$ with probability at least $p_k - p_{k-1}$. In words, verifying each agent that declares $1$ for outcome $j$ with probability at least $p_k - p_{k-1}$ is a sufficient condition for truthfulness.

Therefore, for any fixed $\vec{x}_{-j}$, when $k$ agents declare $1$ for outcome $j$, we need an expected verification of at least $k (p_k - p_{k-1})$ only for agents in group $j$. Assuming truthful reporting and taking the expectation over the number of agents in group $j$ with valuation $1$, we obtain that an expected verification of:
\begin{align*}
 \sum_{k = 0}^{\nu} 2^{-(k+1)} k (p_k - p_{k-1}) & \geq
 \sum_{k = 0}^{\nu} 2^{-(k+1)} p_k (k - (k+1)/2) \\ &=
 \frac{1}{2} \sum_{k = 0}^{\nu} 2^{-(k+1)} p_k k -
 \frac{1}{2} \sum_{k = 0}^{\nu} 2^{-(k+1)} p_k \\ &=
 \Exp[\mbox{welfare}_j\,|\, \vec{x}_{-j}]/2 -
 \Prob[\mbox{outcome $j$}\,|\, \vec{x}_{-j}]/2
\end{align*}
Note that the last equality ignores a negligible term of at most $\nu\delta$. Hence, the expected verification for agents in group $j$ is at least half the expected social welfare of the mechanism from group $j$, conditional on $\vec{x}_{-j}$, minus half the probability of outcome $j$ in $F$, conditional on $\vec{x}_{-j}$.

Removing the conditioning on $\vec{x}_{-j}$ and summing up over all groups $j$, we obtain that the expected verification is at least $\Exp[\mbox{welfare}]/2 - 1/2$ (with expectations taken over the random choices of the mechanism and over the random instances). Since such instances have a maximum social welfare of $\Theta(\log m)$ with high probability and since the mechanism achieves a constant approximation ratio, the mechanism $F$ on a random instance requires an expected verification of $\Omega(\log m)$ (the expectation is taken over the random choices of the mechanism) with high probability (the high probability refers to the selection of the instance).
\qed

\begin{remark}
In the proof of Theorem~\ref{th:lower_bound}, we assume for simplicity that $\nu$ is large and $\delta$ is tiny, so that $\nu\delta$ is negligible. In fact, it suffices to use $\nu = \log m/\eps$, for some small enough constant $\eps > 0$, and $\delta = 1 / \log m$. The proof is essentially identical, only the final calculations change. However, since now the total number of agents is $n = m\log m / \eps$, and $\log n = \Theta(\log m)$, we obtain that even if verification is quantified wrt. the number $n$ of agents, we still need logarithmic verification. 
\end{remark}

\subsection{Lower Bound on Truthful Mechanisms without Verification}
\label{s:app:gen_lower_bound}

  In this section we see that any truthful mechanism without verification achieves the same social welfare as the constant uniform allocation rule. This demonstrates the importance and the quality of our results in a fully general setting such as utilitarian voting. The following theorem is true even for partial allocation rules which means that even this opportunity without verification is useless.

\begin{theorem}
  Any randomized truthful mechanism without money and without verification (full or partial allocation) has approximation ratio to the optimal social welfare at most $m^{-1}$ if we have large enough number of agents.
\end{theorem}

\begin{proof}
  Let $f$ be a truthful allocation rule. Consider the instance where we have $m$ different outcomes and $m$ single-minded agents. Each agent $i$ has valuation one for the outcome $i$ (i.e. $x_i(i) = 1$) and zero for the
others (i.e. $x_i(j) = 0$ for $j \neq i$). Let $p_i$ be the probability of the outcome $i$ in this instance according to $f$ (i.e. $p_i = f_i(\vec{x})$) and let $j = \arg \min_{i} p_i$. Since $\sum_i p_i \le 1$ we have that
$p_j \le m^{-1}$. Lets assume now that the valuation of agent $j$ for the outcome $j$ increases and all the other valuations remain the same. We define
$p_j(x_j(j)) = f_j(\vec{x}_{-j}, \vec{x}_j)$, where $\vec{x}_j = (0, \dots, 0, x_j(j), 0, \dots, 0)$. Since $f$ is truthful we have that $p_j$ should be non decreasing function of $x_j(j)$ otherwise $f$ does not satisfy
monotonicity and therefore is not truthfully implementable even with monetary transfers \cite{AK08}. We prove that $p_j(x_j(j))$ is a constant function as $x_j(j)$ increases. Lets assume the opposite, then there exists an
$L \ge 1$ such that $p_j(L) > p_j$, because $p_j$ is non decreasing. Then agent $j$ has more expected utility if he reports $L$ instead of his real valuation which is $1$ and so $f$ is not truthful. Therefore
$p_j(x_j(j)) = p_j$ for all $x_j(j)$. Now consider the same instance as before with the difference that $x_j(j)$ is much larger. The optimal social welfare is $x_j(j)$ whereas $f$ gets welfare at most
$1 + x_j(j) \cdot m^{-1}$ where the first term comes from the total valuation of all the agents but $j$ and the second term from the fact that $p_j \le m^{-1}$. Finally the approximation ratio of $f$ as $x_j(j)$ goes to
infinity is
\[\lim_{x_j(j) \rightarrow \infty} \frac{1 + x_j(j) \cdot m^{-1}}{x_j(j)} = m^{-1}\]
\qed\end{proof}

  We notice that an approximation ratio of $m^{-1}$ is achievable by the mechanism that select an outcome uniformly at random without taking into account the valuations of the agent. So the uniform mechanism is worst-case
optimal for truthfully maximizing welfare without money in the utilitarian voting setting.

%% file: appendix_impossibility.tex
\subsection{Truthfulness and Low Verification Imply Continuity: The Proof of Lemma~\ref{l:discotruth}}
\label{s:app:discotruth}

Since all our mechanisms and allocation rules are strongly anonymous, throughout this section, we refer to strongly anonymous and scale invariant mechanisms / rules just as scale invariant mechanisms / rules, for simplicity and brevity. Moreover, since we focus on strongly anonymous mechanisms / allocation rules, we always consider the weight vector $\vec{w}(\matr{x}) \equiv \matr{x}$ of the outcomes induced by a valuation profile $\matr{x}$, instead of the valuation profile $\matr{x}$. 

Before we start with the proof, we observe that, by scale invariance, any allocation rule $f$ cannot be continuous at $\vec{0}$, since this would imply that $f$ is a constant allocation that does not depend on the input (to see this, start from any weight vector $\vec{w}$ and choose an scaling factor $\alpha$ that tends to $0$). For this and for similar technical reasons, in this and in the following sections, we restrict our attention to weight vectors $\vec{w}$ with strictly positive value in each coordinate. Thus, we let $\reals_+ = \{ x \in \reals : x > 0\}$ and focus on mechanisms / allocation rules restricted to the domain $\reals_+^m$.
Moreover, we prove some of our results for the general case that the image set of the mechanism / allocation rule is not restricted to the simplex $\Delta(O)$, but it is the more general $\reals_+^m$.

Before the proof of Lemma~\ref{l:discotruth}, we recall the definition of continuity of a real multivariable function. A function $f : \reals_+^m \rightarrow \reals_+^m$ is \emph{continuous} if for every $\vec{w} \in \reals_+^m$
\[ \lim\limits_{\vec{v} \rightarrow \vec{w}} f(\vec{v}) = f(\vec{w}) \]
Moreover, given a vector $\vec{w} \in \reals_+^m$, a multiset $p = \{ \vec{v}_1, \ldots, \vec{v}_n\}$, with each $\vec{v}_i \in \reals_+^m$, is called a \emph{partition} of $\vec{w}$ if $\vec{w} = \sum_{\vec{v}_i \in p} \vec{v}_i$. We let $P(\vec{w})$ denote the set of all possible partitions of a vector $\vec{w}$.

\begin{proof}[of Lemma~\ref{l:discotruth}]
We first show that if $f$ has a discontinuity, there are $\Omega(n)$ agents that have a very small valuation $\delta > 0$ and can change the allocation by a constant factor, independent of $n$ and $\delta$.

To this end, we let $\vec{w}$ be the point of discontinuity of $f$. So $\lim\limits_{\vec{v} \rightarrow \vec{w}} f(\vec{v})$ either does not exist or it is different from $f(\vec{w})$. In both cases, by the definition of continuity,
there exists an $\epsilon > 0$ such that for every $\delta > 0$ there is a point $\vec{v}_\delta$ such that $|\vec{v}_{\delta} - \vec{w}| \le \delta$ and $|f(\vec{v}_{\delta}) - f(\vec{w})| \ge \epsilon$.

\begin{claim}
For every $\delta > 0$, there is a $\vec{v}_\delta$ with the above properties and $\vec{v}_\delta \le \vec{w}$.
\end{claim}

\begin{proof}[of the claim]
Let $C = \min_j w_j$ and $D = \max_j w_j$. Then we take a $\vec{v}_{\delta '}$ with
$\delta' \le D \delta / ( m (C + D) )$. Now we define $\vec{z} = \vec{v}_{\delta'}/(1 + \frac{\delta'}{C})$.
If $(v_{\delta '})_j \le w_j$, then $z_j \le w_j$. On the other hand, if $(v_{\delta '})_j \ge w_j$, then, since $\vec{w} \in \reals_+^m$ and $w_j \neq 0$,
\[ (v_{\delta '})_j - w_j \le \delta' \Rightarrow \frac{(v_{\delta '})_j}{1 + \delta'/w_j} \le w_j \text{\,,\ \ but\ \  } z_j = \frac{(v_{\delta '})_j}{1 + \delta'/C} \le \frac{(v_{\delta '})_j}{1 + \delta'/w_j} \Rightarrow z_j \le w_j \]

Now we bound $|\vec{z} - \vec{w}|$. If $(y_{\delta '})_j \ge w_j$ then
$w_j - z_j \le w_j - w_j / \left( 1 + \frac{\delta'}{C} \right) = w_j \left( \frac{\delta'}{C + \delta'} \right)$.
Otherwise, by the definition of $(v_{\delta '})_j$ we have that $(v_{\delta '})_j \ge w_j - \delta'$ and therefore
$w_j - z_j \le w_j - (w_j - \delta') / \left( 1 + \frac{\delta'}{C} \right)$.
Hence, in any case, we have that
\[ w_j - z_j \le w_j \left( \frac{\delta'}{C + \delta'} \right) + \delta' \left( \frac{C}{C + \delta'} \right) \le D \left( \frac{\delta'}{C + \delta'} \right) + \delta' \left( \frac{C}{C + \delta'} \right) \]
After summing up for all $j$, we get
\[ |\vec{w} - \vec{z}| \le m \delta' \left( \frac{D + C}{C + \delta'} \right) \le m \left( \frac{D + C}{C} \right) \delta' \le \delta\,,\]
where the last inequality follows from the definition of $\delta'$. Now, since $f$ is scale invariant, we know that
$f(\vec{v}_{\delta '}) = f(\vec{z})$. Therefore, $|f(\vec{z}) - f(\vec{w})| \ge \epsilon$ and
$|\vec{z} - \vec{w}| \le \delta$ and $\vec{z} \le \vec{w}$. Since we can do so for every $\delta > 0$, we have found the required $\vec{v}_\delta = \vec{z}$.
\qed\end{proof}

Now since $|f(\vec{v}_{\delta}) - f(\vec{w})| \ge \epsilon$ there exists $j \in [m]$ such that for
$\eps = \epsilon / m$ we have $|f_j(\vec{w}) - f_j(\vec{v}_\delta)| \ge \eps$. This means that either
$f_j(\vec{w}) - f_j(\vec{v}_\delta) \ge \eps$ or $f_j(\vec{v}_\delta) - f_j(\vec{w}) \ge \eps$.

Assume $f_j(\vec{w}) - f_j(\vec{v}_\delta) \ge \eps$. Then, $f_j(\vec{w}) \neq 0$. Let $\vec{\beta}_\delta = \vec{w} - \vec{v}_\delta$, we have that $|\vec{\beta}_\delta| \le \delta$ and $\vec{\beta}_\delta \ge 0$ because of the claim and therefore $\vec{\beta}_\delta$ can be the valuation vector of an
agent. Also let $B = \max_j (\vec{\beta}_\delta)_j$ and $\vec{z}_{\delta} = \vec{\beta}_\delta + \delta \vec{e}_j$
where $\vec{e}_j$ is the unit vector to the direction $j$.

Now let $\alpha = 1 - \frac{\eps}{f_j(\vec{w})}$. Then, because of the discontinuity, we have that $f_j(\vec{w}) - f_j(\vec{v}_\delta) \ge \eps$. Therefore.
\begin{equation}
\label{eq:frac_of_outcomes_1}
  \frac{f_j(\vec{v}_\delta)}{f_j(\vec{w})} \le \alpha
\end{equation}
where $\alpha$ is independent of $\delta$. We let now 
\[ k = \min_{j : (\vec{z}_{\delta})_j > 0} \left\{ \floor{\frac{w_j}{(\vec{z}_{\delta})_j}} \right\} 
\mbox{\ \ \ and\ \ \ } \gamma_j = w_j - k\cdot (\vec{z}_{\delta})_j \] 
We consider the following partition of $\vec{w}$:
$p = \{\underbrace{\vec{z}_{\delta}, \dots, \vec{z}_{\delta}}_{k \text{ times}}, \vec{\gamma}\}$.
So, we have $k$ agents with valuation $\vec{z}_{\delta}$. 

Next, we observe that any of these $k$ agents can significantly (and profitably) change the probability distribution of $f$ by a slight deviation from her true valuation. Therefore, any truthful extension $F$ of $f$ must verify each of these $k$ agents with a probability at least $\zeta$, due to truthfulness, where $\zeta$ is a constant that does not depend on $k$ and $\delta$. Therefore, the expected verification of $F$ is at least linear in the number of agents.

Specifically, let us assume that one of these $k$ agents $i$ has true valuation $\delta \vec{e}_j$, where $\vec{e}_j$ is the unit vector with a single $1$ in its $j$-th coordinate.
Then, she prefers $f(\vec{w})$ to $f(\vec{v}_\delta)$. So, if in this instance, the mechanism $F$ does not verify her, she will misreport $\vec{z}_\delta$ instead of $\delta \vec{e}_j$ to get $f(\vec{w})$.
So, for each agent $i$, we let
\[ vr_i = \Prob[\mbox{$F$ verifies $i$ when reported valuations are as in $p$}\,]\]
Then, if the true valuation of the agent $i$ is $\delta \vec{e}_j$ and she reports $\vec{z}_\delta$, while all other agents report truthfully, the utility of $i$ is at least $\delta \cdot (1 - vr_i) \cdot f_j(\vec{w})$. For this lower bound on $i$'s utility, we assume that if $F$ verifies $i$, agent $i$ gets utility $0$.
On the other hand, if agent $i$ reports her true valuation $\delta \vec{e}_j$, she gets utility $\delta f_j(\vec{v}_\delta)$. But since $F$ is truthful, we have that
\[  \delta \cdot (1 - vr_i) \cdot f_j(\vec{w}) \le \delta \cdot f_j(\vec{v}_\delta) \Rightarrow (1 - vr_i) \le \frac{f_j(\vec{v}_\delta)}{f_j(\vec{w})} \overset{(\ref{eq:frac_of_outcomes_1})}{\Rightarrow} 
(1 - vr_i) \le \alpha \Rightarrow 1 - \alpha \le vr_i\]
Since $\alpha < 1$, by definition, we let $\zeta = 1 - \alpha$. Hence, $vr_i \geq \zeta$ and $\zeta$ is a positive constant that does not depend on $\delta$.
Therefore the expected verification of $F$ is
\[ \ExpVer(F) = \sum_i vr_i \ge k \cdot \zeta = \Omega(n) = 
%\Omega(1 / \beta) =
\Omega(1 / \delta) \]
We note that this holds for any $\delta > 0$ and that the number of agents $n$ can be arbitrary large.

Next, we consider the case where $f_j(\vec{v}_\delta) - f_j(\vec{w}) \ge \eps$.  We define $\vec{\beta}_\delta$, $B$ as before. Also, we let $L = \norm{f(\vec{w})}_{\infty}$ and
$\vec{z}_{\delta} = (2 / \eps) m L B \vec{e}_j$. The intuition and the basic steps of the proof are very similar to those of the proof above. However, the technical details are different, we present the detailed argument below.

We assume, for simplicity, that $f_j(\vec{w}) > 0$ and let $\alpha = 1 + \frac{\eps}{f_j(\vec{w})}$ (we also derive below a similar lower bound on the verification probability for the case where $f_j(\vec{w}) = 0$). Then, because of the discontinuity of $f$, we have that $f_j(\vec{v}_\delta) - f_j(\vec{w}) \ge \eps$. Therefore,
\begin{equation}
\label{eq:frac_of_outcomes_2}
  \frac{f_j(\vec{v}_\delta)}{f_j(\vec{w})} \ge \alpha\,,
\end{equation}
where $\alpha$ is independent of $\delta$. Now, we let
\[ k = \floor{\frac{(\vec{v}_\delta)_j}{(\vec{z}_{\delta})_j}} \ge \floor{\frac{w_j - \delta}{(2/\eps) m L B}} = \Omega(1/\delta) \]
For the last equality, we use that $\eps, m, L, B$ are independent of $\delta$. We also let $\gamma_j = (\vec{v}_\delta)_j - k (\vec{z}_{\delta})_j$.
Then, we consider the following partition of $\vec{v}_\delta$:
$p = \{\underbrace{\vec{z}_{\delta}, \dots, \vec{z}_{\delta}}_{k \text{ times}}, \vec{\gamma}\}$.
So, we have $k$ agents with valuation $\vec{z}_{\delta}$. 

Let us assume that one of these agents $i$ has true valuation
$\vec{\beta}_\delta + (2 / \eps) m L B \vec{e}_j$.
Then, $i$ prefers $f(\vec{w})$ %$f(\vec{y}_\delta)$
to $f(\vec{v}_\delta)$ because the $j$-th coordinate of her valuation dominates the others. So, if in this instance, the mechanism $F$ does not verify her, agent $i$ would report $\vec{z}_\delta$, instead of $\vec{\beta}_\delta + (2 / \eps) m L B \vec{e}_j$, in order to get $f(\vec{w})$. So, using our definition of the verification probability $vr_i$, if the true valuation of the agent $i$ is $\vec{\beta}_\delta + (2 / \eps) m L B \vec{e}_j$ and all other agents report truthfully, the utility of agent $i$ is at least
  \(  (2 / \eps) m L B \cdot (1 - vr_i) \cdot f_j(\vec{v}_\delta) \).
For this lower bound on $i$'s utility, we assume that if $F$ verifies $i$, agent $i$ gets utility $0$.
If the agent $i$ reports truthfully, she gets utility
\[ \sum_{l} (\vec{\beta}_\delta)_l f_l(\vec{w}) + (2 / \eps) m L B f_j(\vec{w}) \le m L B + (2 / \eps) m L B f_j(\vec{w})\,, \]
where the inequality follows from the definition of $B$ and $L$. But since $F$ is truthful, we have that
\[ (2 / \eps) m L B \cdot (1 - vr_i) \cdot f_j(\vec{v}_\delta) \le m L B + (2 / \eps) m L B f_j(\vec{w}) \Rightarrow \]
\[ \Rightarrow \frac{2}{\eps}\cdot (1 - vr_i) \cdot f_j(\vec{v}_\delta) \le 1 + \frac{2}{\eps} \cdot f_j(\vec{w}) \overset{(\ref{eq:frac_of_outcomes_2})}{\Rightarrow} \frac{2}{\eps} \cdot (1 - vr_i)\,\alpha \le \frac{1}{f_j(\vec{w})} + \frac{2}{\eps} \Rightarrow\]
\[ \Rightarrow \frac{2}{\eps} \cdot (\alpha - 1) - \frac{1}{f_j(\vec{w})} \le \alpha \,vr_i \Rightarrow \frac{2}{\eps}\cdot \frac{\eps}{f_j(\vec{w})} - \frac{1}{f_j(\vec{w})} \le \alpha\,vr_i \Rightarrow \frac{1}{f_j(\vec{w}) + \eps} \le vr_i \]

The last implications follow from calculations that use the definition of $\alpha$. We let $\zeta = 1/(f_j(\vec{w}) + \eps)$. Thus, $vr_i \ge \zeta$ and $\zeta$ is a positive constant that does not depend on $\delta$. 

In case where $f_j(\vec{w}) = 0$, the truthfulness condition becomes
\[ (1 - vr_i) (2\eps) m L B f_j(\vec{v}_\delta) \le m L B \Rightarrow (1 - vr_i) (2/\eps) f_j(\vec{v}_\delta) \le 1 \Rightarrow vr_i \geq 1/2 \]
The last implication above follows from the fact that $f_j(\vec{v}_\delta) - f_j(\vec{w}) \ge \eps$ $\Rightarrow$ $f_j(\vec{v}_\delta) \ge \eps$. 
So, in this case, we let $\zeta = 1/2$. As before, the expected verification of $F$ is
\[ \ExpVer(F) = \sum_i vr_i \ge k \cdot \zeta = \Omega(n) = 
%\Omega(1 / \beta) = 
\Omega(1 / \delta) \]
Note that this holds for any $\delta$ and that the number of agents $n$ can be arbitrary large.
\qed\end{proof}

\subsection{Continuity Implies Participation}
\label{s:app:cont-parti}

In this section, we show that participation is a necessary condition for any continuous scale invariant and strongly anonymous full allocation rule $f$ that can be extended to a truthful mechanism $F$ with selective verification. By Lemma~\ref{l:discotruth}, this also holds for any scale invariant and strongly anonymous full allocation rule $f$ that can be extended to a truthful mechanism $F$ with selective verification of $o(n)$ agents. 

\begin{lemma}\label{l:cont-parti}
Let $f$ be any continuous scale invariant and strongly anonymous full allocation rule. If $f$ has a truthful extension $F$, then $f$ satisfies participation.
\end{lemma}

\begin{proof}
We first show that if any agent $i$ declares valuation $\vec{0}$ for all outcomes, the probability distribution of the truthful extension $F$ of $f$ is not affected by whether $\vec{0}$ is the true valuation of $i$ or a misreport. To this end, we fix the valuation profile $\matr{x}_{-i}$ of all other agents and let $\vec{w}_{-i}$ be the outcome weight vector induced by $\matr{x}_{-i}$. We next show that for any agent $i$,
\begin{equation}\label{eq:equality_parti}
  F((\vec{x}_{-i}, \vec{0}), (\vec{1}_{-i}, 1)) = F((\vec{x}_{-i}, \vec{0}), (\vec{1}_{-i}, 0))
\end{equation}

Let us assume that $F((\vec{x}_{-i}, \vec{0}), (\vec{1}_{-i}, 1)) \neq F((\vec{x}_{-i}, \vec{0}), (\vec{1}_{-i}, 0))$.
Then, since $f$ achieves full allocation, 
$|F((\vec{x}_{-i}, \vec{0}), (\vec{1}_{-i}, 1))| = 1$ and
$|F((\vec{x}_{-i}, \vec{0}), (\vec{1}_{-i}, 0))| = 1$. Therefore, there is an outcome $j$ such that
$F_j((\vec{x}_{-i}, \vec{0}), (\vec{1}_{-i}, 0)) > F_j((\vec{x}_{-i}, \vec{0}), (\vec{1}_{-i}, 1))$.
Let
$\gamma \equiv F_j((\vec{x}_{-i}, \vec{0}), (\vec{1}_{-i}, 0)) - F_j((\vec{x}_{-i}, \vec{0}), (\vec{1}_{-i}, 1))$ be the difference in the probability of outcome $j$ in $F$ when $\vec{0}$ is a misreport of agent $i$ and when agent $i$ truthfully declares $\vec{0}$.

Let $\vec{e}_j$ be the unit vector with a single $1$ in its $j$-th coordinate. Since $F$ is an extension of $f$ and since $f$ is strongly anonymous, (i) $F((\vec{x}_{-i}, \vec{0}), (\vec{1}_{-i}, 1)) = f(\vec{w}_{-i})$\,; and (ii) for any $h > 0$,
$F((\vec{x}_{-i}, h \vec{e}_j), (\vec{1}_{-i}, 1)) = f( \vec{w}_{-i} + h \vec{e}_j )$. Since $f$ is continuous,
  \[ \lim\limits_{h \rightarrow 0} f(\vec{w}_{-i} + h \vec{e}_j) = f(\vec{w}_{-i}) = F((\vec{x}_{-i}, \vec{0}), (\vec{1}_{-i}, 1)) \]

Intuitively, the continuity of $f$ implies that $F_j((\vec{x}_{-i}, \delta \vec{e}_j), (\vec{1}_{-i}, 1)) \approx F_j((\vec{x}_{-i}, \vec{0}), (\vec{1}_{-i}, 1))$, if $\delta$ is small enough. Moreover, truthfulness implies that $F_j((\vec{x}_{-i}, \delta \vec{e}_j), (\vec{1}_{-i}, 1)) \geq F_j((\vec{x}_{-i}, \vec{0}), (\vec{1}_{-i}, 0))$. But these contradict our assumption that $F_j((\vec{x}_{-i}, \vec{0}), (\vec{1}_{-i}, 0)) > F_j((\vec{x}_{-i}, \vec{0}), (\vec{1}_{-i}, 1))$. Thus, we obtain (\ref{eq:equality_parti}).

Let us now formalize the intuition above. By continuity, there are an $\e \in (0, \gamma)$ and a $\delta > 0$ such that
\[ |f(\vec{w}_{-i} + \delta \vec{e}_j) - f(\vec{w}_{-i})| \le \e \Rightarrow
|F((\vec{x}_{-i}, \delta \vec{e}_j), (\vec{1}_{-i}, 1)) - F((\vec{x}_{-i}, \vec{0}), (\vec{1}_{-i}, 1))| \le \e \]
This implies that
$|F_j((\vec{x}_{-i}, \delta \vec{e}_j), (\vec{1}_{-i}, 1)) - F_j((\vec{x}_{-i}, \vec{0}), (\vec{1}_{-i}, 1))| \le \e < \gamma$.
Therefore, by the definition of $\gamma$, we obtain that
\begin{equation}
\label{eq:contr_continuity}
F_j((\vec{x}_{-i}, \vec{0}), (\vec{1}_{-i}, 0)) > F_j((\vec{x}_{-i}, \delta \vec{e}_j), (\vec{1}_{-i}, 1))
\end{equation}

Applying truthfulness when agent $i$ has true valuation $\delta \vec{e}_j$, we obtain that
  \[ \delta \vec{e}_j \cdot F((\vec{x}_{-i}, \delta \vec{e}_j), (\vec{1}_{-i}, 1)) \ge \delta \vec{e}_j \cdot F((\vec{x}_{-i}, \vec{0}), (\vec{1}_{-i}, 0)) \]
Since $\delta > 0$, we get that
$F_j((\vec{x}_{-i}, \delta \vec{e}_j), (\vec{1}_{-i}, 1)) \ge F_j((\vec{x}_{-i}, \vec{0}), (\vec{1}_{-i}, 0))$ which
contradicts (\ref{eq:contr_continuity}). So, we obtain that (\ref{eq:equality_parti}) holds for any agent $i$.

Since $F$ is an extension of $f$ and since $f$ is strongly anonymous, 
$F((\vec{x}_{-i}, \vec{0}), (\vec{1}_{-i}, 1)) = f(\vec{w}_{-i})$ and
$F((\vec{x}_{-i}, \vec{x}_i), (\vec{1}_{-i}, 1)) = f(\vec{w}_{-i} + \vec{x}_i)$, for the valuation $\vec{x}_i$ of agent $i$.
Using  (\ref{eq:equality_parti}), we obtain that
$f(\vec{w}_{-i}) = F((\vec{x}_{-i}, \vec{0}), (\vec{1}_{-i}, 0))$. Now, using truthfulness, we conclude that:
\[    \vec{x}_i \cdot F((\vec{x}_{-i}, \vec{x}_i), (\vec{1}_{-i}, 1)) \geq \vec{x}_i \cdot F((\vec{x}_{-i}, \vec{0}), (\vec{1}_{-i}, 0)) \Rightarrow
\vec{x}_i \cdot f( \vec{w}_{-i} + \vec{x}_i ) \geq \vec{x}_i \cdot f( \vec{w}_{-i} ) \]
Therefore, $f$ satisfies participation (and thus, $F$ also satisfies participation, since $F$ is an extension of $f$).
\qed \end{proof}

\subsection{Participation implies Discontinuity}
\label{s:app:discont}

Next, we use Lemma~\ref{l:chara_weak_parti} and prove that if a full allocation rule $f$ is strongly anonymous and scale invariant and satisfies participation, then $f$ is either constant (i.e., its allocation is independent of the reported valuations) or has a discontinuity at $\vec{1}$. The intuition is that by Lemma~\ref{l:chara_weak_parti}, any continuous allocation rule $f$ with these properties must satisfy $f(\vec{w}) = \arg \max_{\vec{z} \in \Delta(O)} \vec{w} \cdot \vec{z}$. But, welfare maximizers over $\Delta(O)$, i.e., with full allocation, cannot be continuous (e.g., let $m = 2$ and consider welfare maximization over the unit simplex for weights vectors $(1, 1+\e)$ and $(1, 1-\e)$). The proof of the following formalizes this intuition.

\begin{lemma}\label{l:discont}
Let $f$ be any strongly anonymous and scale invariant full allocation rule that satisfies participation. Then, the probability distribution of $f$ is either constant or discontinuous.
\end{lemma}

\begin{proof}
Let $f$ be a strongly anonymous and scale invariant full allocation rule that satisfies participation and is not constant. To reach a contradiction, we assume that $f$ is also continuous.
Then, by Lemma~\ref{l:chara_weak_parti}, the allocation of $f$ can be regarded as the solution to an optimization problem where the feasible region is given by a set $Z$ and the direction of the objective is given by the input $\vec{w}$ to $f$. Thus $f$ depends only on the direction of $\vec{w}$, which is expected, due to scale invariance. Since the allocation of $f$ satisfies $\sum_j f_j(\vec{w}) = 1$, $f$ is a solution to an optimization problem on the hyperplane
$p_1 : \sum_j f_j = 1$ with the feasible region given by $Z$ and the direction of
the objective is given by the projection of $\vec{w}$ to the hyperplane $p_1$. We observe that any vector $\vec{w}$ can be written as $\vec{w} = \vec{w}_{p_1} + \vec{w}_{n_1}$, where $\vec{w}_{p_1}$ is the projection of $\vec{w}$ to $p_1$ and $\vec{w}_{n_1}$ is the projection of $\vec{w}$ to a direction perpendicular to $p_1$, which is the direction of $n_1 = (1, 1, \ldots, 1)$. Therefore, since the
optimization problem which determines the value of $f$ is an optimization problem on the $p_1$, the allocation of $f$ on input $\vec{w}$ depends only on the direction of $w_{p_1}$. Hence, for any $a, b > 0$ if $\vec{w}' = a \vec{w}_{p_1} + b \vec{w}_{n_1}$, then $f(\vec{w}) = f(\vec{w}')$
\footnote{Formally, this follows from the description of $f$ as an MIDR allocation rule that satisfies $f(\vec{w}) = \arg \max_{\vec{z} \in Z} \vec{w} \cdot \vec{z}$, by Lemma~\ref{l:chara_weak_parti}, and from the scale invariance of $f$.}.

Now since $f$ is not constant, there exist two points $\vec{w}, \vec{v} \in \reals_+^m$ such that $f(\vec{w}) \neq f(\vec{v})$. From the discussion above, we have that $\vec{w} = \vec{w}_{p_1} + \vec{w}_{n_1}$ and that
$\vec{v} = \vec{v}_{p_1} + \vec{v}_{n_1}$.
We next show that $f$ is left discontinuous at the point
$\vec{w}_{n_1} + \vec{v}_{n_1}$. We take the point $\vec{w}' = a \vec{w}_{p_1} + \vec{w}_{n_1}$, with $a > 0$ to be chosen so that $\vec{w}' < \vec{w}_{n_1} + \vec{v}_{n_1}$, which means $a \vec{w}_{p_1} < \vec{v}_{n_1}$. This is
always possible because $\vec{v}_{n_1}$ has the form $(r, r, \dots, r)$ and we can assume wlog. that $r > 0$, since $f(\vec{v})$ depends only on $\vec{v}_{p_1}$. Similarly, we define
$\vec{v}' = b \vec{v}_{p_1} + \vec{v}_{n_1}$, with $b > 0$ such that $\vec{v}' < \vec{w}_{n_1} + \vec{v}_{n_1}$.
From the discussion above, we have that $f(\vec{w}') = f(\vec{w})$ and that $f(\vec{v}') = f(\vec{v})$.
We now define
$\vec{w}_{\epsilon} = \epsilon \vec{w}' + (1 - \epsilon) (\vec{w}_{n_1} + \vec{v}_{n_1})$. It is not hard to see that
$\vec{w}_{\epsilon} < \vec{w}_{n_1} + \vec{v}_{n_1}$ and that
$\lim\limits_{\epsilon \rightarrow 0} \vec{w}_{\epsilon} = \vec{w}_{n_1} + \vec{v}_{n_1}$.
Moreover,
\[ \vec{w}_{\epsilon} = \epsilon \vec{w}_{p_1} + \vec{w}_{n_1} + (1 - \epsilon) \vec{v}_{n_1} = \epsilon \vec{w}_{p_1} + c \vec{w}_{n_1} \]
for some $c > 0$. Therefore, $f(\vec{w}_{\epsilon}) = f(\vec{w}') = f(\vec{w})$. In the same way, we define
$\vec{v}_{\delta} = \delta \vec{v}' + (1 - \delta) (\vec{w}_{n_1} + \vec{v}_{n_1})$. Then, we have that
$\vec{v}_{\delta} < \vec{w}_{n_1} + \vec{v}_{n_1}$,
that $\lim\limits_{\delta \rightarrow 0} \vec{v}_{\delta} = \vec{w}_{n_1} + \vec{v}_{n_1}$, and that
$f(\vec{v}_{\delta}) = f(\vec{v}') = f(\vec{v})$. Therefore,
\[\lim\limits_{\epsilon \rightarrow 0} f(\vec{w}_{\epsilon}) = f(\vec{w}) \neq f(\vec{v}) = \lim\limits_{\delta \rightarrow 0} f(\vec{v}_{\delta})\]
Hence, the limit $\lim\limits_{\vec{y} \rightarrow (\vec{w}_{n_1} + \vec{v}_{n_1})^-} f(\vec{y})$ does not exist and $f$ is left discontinuous at $\vec{w}_{n_1} + \vec{v}_{n_1}$. In fact, since $f$ is scale invariant, it is discontinuous at
$(1, 1, \ldots, 1) = \vec{x}^* = (\vec{w}_{n_1} + \vec{v}_{n_1}) / \left|\vec{w}_{n_1} + \vec{v}_{n_1}\right|$.
%$(1, 1, \ldots, 1) = \vec{x}^* = (\vec{w}_{n_1} + \vec{v}_{n_1}) / \norm{\vec{w}_{n_1} + \vec{v}_{n_1}}$.
\qed\end{proof}

%% file: appendix_partial.tex
\subsection{The Analysis of Partial Power}
\label{s:app:partial}

In this section, we prove the main properties of the Partial Power allocation rule. Since Partial Power is a strongly anonymous allocation rule, we restrict our attention to the weight vector $\vec{w} \equiv \vec{w}(\vec{x}) = \sum_{i=1}^n \vec{x}_i$ of the outcomes, instead of the valuation profile $\vec{x}$. For some fixed integers $\ell, r \geq 1$, we let $f(\vec w)$ denote the allocation rule:
$$f^{(\ell,r)}(\vec{w}) = \frac {(1 - 1/r)} {m^{1/(\ell+1)}} \cdot \frac {\vec{w}^{\ell}}  {\norm{\vec{w}^{\ell}}_{1+1/\ell}}$$
Interestingly, the Partial Power allocation rule bears a resemblance to proper scoring rules in \cite{GR07}. For simplicity of notation, whenever we use $\norm{\vec{v}}$, without an index denoting the order of the norm, we refer to the $(\ell+1)/\ell$-th norm $\norm{\vec{v}}_{1+1/\ell}$ of vector $\vec{v}$.

\begin{lemma}\label{l:optimization}
For all integers $\ell, r \geq 1$, the allocation rule $f^{(\ell,r)}(\vec{w})$ is the solution to the optimization problem $\max_{\vec{z} \in Z_{\ell, r}} \vec{w} \cdot \vec{z}$, where
$$Z_{\ell, r} = \left\{\vec{z} \in \reals^m_{\geq 0}\,:\,\norm{\vec{z}}_{1+1/\ell} \le (1 - 1/r) m^{-1/(\ell+1)} \right\}$$
\end{lemma}

\begin{proof}
The range $Z_{\ell, r}$ is smooth and strictly convex. Therefore, the solution to the optimization problem $\max_{\vec{z} \in Z_{\ell, r}} \vec{w} \cdot \vec{z}$ is the extreme point $\vec p \in Z_{\ell, r}$ in the direction of $\vec{w}$. This point $\vec{p}$ satisfies  $\nabla \norm{\vec{p}} = \lambda \vec{w}$ for some $\lambda \in \reals$ which may depend on $\vec{w}$.
Moreover, we observe that $\nabla \norm{\vec{p}} = (1+1/\ell) \vec{p}^{1/\ell} / \norm{\vec{p}}$. Hence, since $\vec{p}$ is the extreme point of $Z_{\ell, r}$ in the direction of $\vec{w}$, we have that $\vec{p} = \lambda \vec{w}^{\ell}$. Also, since $\vec{p}$ is a point on the boundary of $Z_{\ell, r}$, $\norm{\vec{p}} = (1 - 1/r) m^{-1/(\ell+1)}$, which gives the desired closed form for $f^{({\ell, r})}(\vec{w})$.
\qed\end{proof}

Lemma~\ref{l:optimization} and Lemma~\ref{l:participation} imply that the Partial Power allocation rule is MIDR and satisfies participation. The following provides an upper bound on the approximation ratio of Partial Power. The intuition behind the approximation ratio is that due to the definition of the range $Z_{\ell,r}$, the worst case for Partial Power happens when $\vec{w}$ has a single $1$ and all other components $0$. Then, the maximum social welfare is $1$, while the expected social welfare of the partial allocation achieved by partial power is $\frac{1-1/r}{m^{1/(\ell+1)}}$.

\begin{lemma}\label{l:appro_ppower}
For all integers $\ell, r \geq 1$, the allocation rule $f^{(\ell,r)}(\vec{w})$
achieves an approximation ratio of $(1 - 1/r)m^{- 1 / (\ell + 1)}$ for the social welfare.
\end{lemma}

\begin{proof}
First, we recall that we use $\norm{\vec{w}}$ to denote the $(\ell+1)/\ell$-th norm $\norm{\vec{w}}_{1+1/\ell}$ of vector $\vec{w}$. The optimal social welfare is $\max_{j \in O} w_j = \norm{\vec{w}}_{\infty}$. Thus, the approximation ratio of $f^{(\ell,r)}(\vec{w})$ is
\[ \frac{\vec{w} \cdot f^{(\ell,r)}(\vec{w})} {\norm{\vec{w}}_{\infty}} = \frac{(1 - 1/r) \,\vec{w} \cdot \vec{w}^{\ell}}{m^{1/(\ell + 1)}\,\norm{\vec{w}}_{\infty} \norm{\vec{w}^{\ell}}} \]

We observe that $\vec{w} \cdot \vec{w}^{\ell} = \sum_{j\in O} w_j^{\ell + 1} = \norm{\vec{w}^{\ell}}^{(\ell + 1)/\ell}$\,. Therefore,
\[ \frac{\vec{w} \cdot f^{(\ell,r)}(\vec{w})} {\norm{\vec{w}}_{\infty}} = \frac{(1 - 1/r)\,\norm{\vec{w}^{\ell}}^{1/\ell}}{m^{1/(\ell + 1)}\, \norm{\vec{w}}_{\infty}} \]
Now, we observe that for every $\ell \geq 1$,
\[ \norm{\vec{w}^{\ell}}^{1/\ell} =
\left(\sum_{j\in O} \left(w_j^{\ell}\right)^{\frac{\ell+1}{\ell}}\right)%
^{\frac{\ell}{\ell+1}\cdot\frac{1}{\ell}} =
\left(\sum_{j\in O} w_j^{\ell+1}\right)^{\ell+1} = \norm{\vec{w}}_{\ell + 1} \geq
\norm{\vec{w}}_{\infty} \]
Therefore the approximation ratio of Partial Power is $(1 - 1/r) m^{-1/(\ell + 1)}$\,.
\qed\end{proof}

\subsection{The Implementation of Partial Power by Mechanism~\ref{alg:partial_power}}
\label{s:app:ppower}

The analysis in Section~\ref{s:app:partial} focuses on the properties of the Partial Power allocation rule. In this section, we establish that the implementation of Partial Power by Mechanism~\ref{alg:partial_power} with selective verification is robust. Throughout this section, we fix a valuation profile $\vec{x}$ and a verification vector $\vec{s}$. We always write $\PPow^{\ell,r}$, instead of $\PPow^{\ell,r}(\matr{x}, \vec{s})$, for simplicity. Moreover, we let $\vec{w} \equiv \vec{w}(\vec{x}) = \sum_{i=1}^n \vec{x}_i$ denote the outcome weight vector induced by $\vec{x}$.

In the next three lemmas, we establish the feasibility of $\PPow^{\ell,r}$, for any integers $\ell, r \geq 1$, by showing that whenever $\PPow^{\ell,r}$ results in $\bot$, we can allocate a probability $p_j$ each to outcome $j$ so that $\PPow^{\ell,r}$ is robust. We recall that the probability $p_j$ of each outcome $j \in O$ is:
\begin{equation}\label{eq:ppower_pj}
 p_j = \frac{f^{(\ell,r)}_j(\vec{w}_T) - \Prob[\PPow^{\ell,r} = j\,|\,\PPow^{\ell,r} \neq \bot] \, \Prob[\PPow^{\ell,r} \neq \bot]}{\Prob[\PPow^{\ell,r} = \bot]} \,,
\end{equation}
where $T$ is the set of truthful agents in $\matr{x}$. Note that for each outcome $j$,
the probability $p_j$ is defined so that when the truthful players have weight $\vec{w}_T$, the unconditional probability of outcome $j$ in $\PPow^{\ell,r}$ is
\[
  \Prob[\PPow^{\ell,r} = j\,|\,\PPow^{\ell,r} \neq \bot] \,\Prob[\PPow^{\ell,r} \neq \bot] + p_j\,{\Prob[\PPow^{\ell,r} = \bot]} = {f^{(\ell,r)}_j(\vec w_T)} \]
Therefore, with these probabilities $p_j$ in the action $\bot$, $\PPow^{\ell,r}$ becomes robust.

To establish that these probabilities are indeed feasible, we show that for all outcomes $j$, $p_j \geq 0$, and that $\sum_{j \in O} p_j \leq 1$. The non-negativity of $p_j$'s follows rather easily from the fact that excluding some misreporting agents from $\PPow^{\ell,r}$ can only increase the probability that $\PPow^{\ell,r}$ results in outcome $j$ (see also Lemma~\ref{l:pj_positive}). The upper bound on $\sum_{j \in O} p_j$ is more difficult to establish. The intuition is that the additional verification in steps~1-4 of $\PPow^{\ell,r}$ makes the probability that some misreporting agent is caught large enough. Then, the required ``corrections'' in the unconditional probability distribution of $\PPow^{\ell,r}$, which are implemented by the probabilities $p_j$, should not be large. Hence, the sum of $p_j$'s can be upper bounded by $1$ (see also Lemma~\ref{l:sump}).

To simplify the notation, we let $\bot_1$ denote the case where $\PPow^{\ell,r}$ results in $\bot$ at step~4, and $\bot_2$ for the case where $\PPow^{\ell,r}$ results in $\bot$ at step~9. Using this notation, we note that the probability of not revealing any liars in steps~1-4 is
\[
 \Prob\left[ \PPow^{\ell,r} \neq \bot_1 \right] =
  \left( \frac{\left|\vec{w}_T^{\ell + 1}\right|}
  {\left|\vec{w}^{\ell + 1}\right|} \right)^{\!\!r}
\]
Moreover, for the steps~5-10 of $\PPow^{\ell,r}$, we have that the probability of outcome $j$ is:
$$\Prob[\PPow^{\ell,r} = j\,|\,\PPow^{\ell,r} \neq \bot_1] = \frac{1 - 1/r}{m^{1/(\ell + 1)}} \cdot \frac{(w_T(j))^\ell}{\norm{\vec{w}^{\ell}}_{1+1/\ell}}\,,$$
because outcome $j$ is selected if and only if the term selected in step~6 contains only truthful agents.

Using these properties, we can show that $p_j \geq 0$. This is an immediate consequence of the following.

\begin{lemma}\label{l:pj_positive}
For all integers $\ell, r \geq 1$ and all outcomes $j$,
\[
 f^{(\ell,r)}_j(\vec{w}_T) \geq
 \Prob[\PPow^{\ell,r} = j\,|\,\PPow^{\ell,r} \neq \bot] \,
 \Prob[\PPow^{\ell,r} \neq \bot]
\]
\end{lemma}

\begin{proof}
We observe that the following holds for the probability on the rhs:
\begin{align*}
\Prob[\PPow^{\ell,r} = j\,|\,\PPow^{\ell,r} \neq \bot] \, \Prob[\PPow^{\ell,r} \neq \bot]
&\le \Prob[\PPow^{\ell,r} = j\,|\,\PPow^{\ell,r} \neq \bot_1] \\
&= \frac{1 - 1/r}{m^{1/(\ell + 1)}} \cdot \frac{(w_T(j))^\ell}{\norm{\vec{w}^{\ell}}_{1+1/\ell}}\\
& \le \frac{1 - 1/r}{m^{1/(\ell + 1)}} \cdot \frac{(w_T(j))^\ell}{\norm{\vec{w}_T^{\ell}}_{1+1/\ell}} = f^{(\ell,r)}_j(\vec{w}_T)
\end{align*}
\qed\end{proof}

To establish that $\sum_j p_j \le 1$, we need the following technical lemma, which holds for all integers $\ell, r \geq 1$. In the following lemma and in its proof, we use $f(\vec{w})$, instead of $f^{(\ell,r)}(\vec{w})$, and $\norm{\vec{v}}$, instead of $\norm{\vec{v}}_{1+1/\ell}$, for simplicity.

\begin{lemma}\label{l:null_bound}
For all vectors $\vec{v}, \vec{w} \in \reals_{\geq 0}^m$
\[
 \abs {f(\vec{w})} - \abs { f(\vec{w} + \vec{v}) } \le
  1 - \frac{\left|\vec{w}^{\ell + 1}\right|}{\left|(\vec{w}+\vec{v})^{\ell + 1}\right|}
\]
\end{lemma}

\begin{proof}
Using the definition of the Partial Power allocation $f(\vec{w})$, we obtain that:
\begin{align*}
\abs {f(\vec{w})} - \abs { f(\vec{w} + \vec{v}) } & = \frac{1 - 1/r}{m^{1/(\ell + 1)}} \left( \frac{\abs{\vec{w}^{\ell}}}{\norm{\vec{w}^{\ell}}} - \frac{\abs{(\vec{w} + \vec{v})^{\ell}}}{\norm{(\vec{w} + \vec{v})^{\ell}}} \right) \\ & =
\frac{1 - 1/r}{m^{1/(\ell + 1)}}  \cdot \frac{\abs{\vec{w}^{\ell}}}{ \norm{\vec{w}^{\ell}}} \left( 1 - \frac{\norm{\vec{w}^{ \ell}}}{\abs{\vec{w}^{ \ell}}} \frac{\abs{(\vec{w} + \vec{v})^{\ell}}}{\norm{(\vec{w} + \vec{v})^{\ell}}} \right)
\end{align*}

Since $\abs{\vec{w}^{\ell}} \leq m^{1/(\ell + 1)} \norm{\vec{w}^{\ell}}$ and $\abs{(\vec{w} + \vec{v})^{\ell}} \geq \abs{\vec{w}^{\ell}}$, we obtain that
\begin{align*}
\abs {f(\vec{w})} - \abs { f(\vec{w} + \vec{v}) } & \le
(1-1/r) \left( 1 - \frac{\norm{\vec{w}^{\ell}}}{\norm{\left(\vec{w} + \vec{v}\right)^{\ell}}} \right) \\
& \le (1-1/r) \left( 1 - \frac{\norm{\vec{w}^{\ell}}^{\frac{\ell + 1}{\ell}}}{\norm{\left(\vec{w} + \vec{v}\right)^{\ell}}^{\frac{\ell + 1}{\ell}}} \right)
 \leq 1 - \frac{\left|\vec{w}^{\ell + 1}\right|}{\left|(\vec{w} + \vec{v})^{\ell + 1}\right|}
\end{align*}
For the last inequality, we use that for the $(\ell+1)/\ell$-th norm, $\norm{\vec{w}^{\ell}}^{(\ell + 1)/\ell} = \left|\vec{w}^{\ell + 1}\right|$.
\qed \end{proof}

We next show that $\sum_{j \in O} p_j \leq 1$, which implies the feasibility of the implementation of Partial Power by Mechanism~\ref{alg:partial_power}. In fact, $\sum_{j \in O} p_j \leq 1$ is an immediate consequence of the following. In the statement of the lemma and in its proof, we use $f(\vec{w})$, instead of $f^{(\ell,r)}(\vec{w})$, and $\norm{\vec{v}}$, instead of $\norm{\vec{v}}_{1+1/\ell}$, for simplicity. Also, we recall that since we regard the probability distribution of the allocation rule $f$ as a vector over outcomes, $|f(\vec{w})| = \sum_{j \in O} f_j(\vec{w})$ is the probability that $f$ results in some outcome in $O$.

\begin{lemma}\label{l:sump}
For all integers $\ell, r \geq 1$, the following holds:
$$ \sum_{j \in O} \left( {f_j(\vec{w}_T) - \Prob[\PPow^{\ell,r} = j ]} \right)  \le \Prob[\PPow^{\ell,r} = \bot] $$
\end{lemma}

\begin{proof}
For brevity, we let $\rho = \frac{\left|\vec{w}_T^{\ell + 1}\right|}{\left|\vec{w}^{\ell + 1}\right|}$ throughout the proof. The desired inequality can be rewritten as:
\[ \Prob[\PPow^{\ell,r} = \bot] + \sum_{j \in O} \Prob[\PPow^{\ell,r} = j ] \ge \abs{f(\vec w_T)} \]

We observe that the lhs of the inequality above is $1 - \Prob[\PPow^{\ell,r} = \mathrm{null} ]$, which is equal to:
\[ 1 - \Prob[\PPow^{\ell,r} \neq \bot_1]\,\Prob\left[ \PPow^{\ell, r}(\matr{x}, \vec{s}) = \mathrm{null}\,|\,\PPow^{\ell,r} \neq \bot_1  \right] = 1 - \rho^r \cdot (1 - \abs{ f(\vec{w})} ) \]
Therefore, it suffices to show that $ 1 - \rho^r \cdot (1 - \abs{ f(\vec{w})} ) - \abs{ f(\vec{w}_T)} \ge 0$.

We observe that Lemma~\ref{l:null_bound} implies that $\abs{ f(\vec{w})} \ge \abs{ f(\vec{w}_T)} - (1-\rho)$. Hence, it suffices to show that:
\begin{align*}
1 - \rho^r \cdot (1 - \abs{ f(\vec{w}_T)} - (1-\rho)) - \abs{ f(\vec{w}_T)} & \ge 0
\Leftrightarrow \\
( 1 - \rho^r ) \cdot (1 - \abs{ f(\vec{w}_T)} ) - \rho^r (1-\rho) & \ge 0
\end{align*}
Since $\abs{ f(\vec{w}_T)} \leq 1 - 1/r$, we obtain that
\[
 \frac{1 - \rho^r}{r} - \rho^r (1-\rho) \ge 0 \Leftrightarrow \\
 \frac { 1 - \rho^r } {1-\rho}  - r \rho^r  \ge 0 \Leftrightarrow \\
 \sum_{k=0}^{r-1} \rho^k   -  r \rho^r  \ge 0
\]
The latter inequality is always true because $\rho^k \ge \rho^r$, for all $0 \le k \le r$, since $\rho \in [0,1]$.
\qed\end{proof}

We have shown that for all integers $\ell, r \geq 1$, $\PPow^{\ell, r}$ is robust. Since the Partial Power allocation satisfies the participation constraint, Lemma~\ref{l:robustness+participation} implies that $\PPow^{\ell, r}$ is truthful. To complete the proof of Theorem~\ref{th:partial-power}, we set $\ell = O(\ln m / \eps)$ and $r = O(1 / \eps)$ and get the desired approximation guarantee.

%% file: appendix_exponential.tex
\begin{algorithm}[t]
\caption{\label{alg:exponential}The Exponential Mechanism $\Expo^\alpha(\vec{x}, \vec{s})$}
\begin{algorithmic}\normalsize
    \State Let $N$ be the set of the remaining agents and let $L \leftarrow \emptyset$
    \State pick an outcome $j \in O$, an integer $\ell \geq 0$
           and a tuple $\vec{t} \in N^{\ell}$ \\
\ \ \ \ \ \ with probability proportional to the value of the term
           $x_{t_1}(j) x_{t_2}(j) \cdots x_{t_\ell}(j) / (\alpha^\ell \ell!)$
    \For{\textbf{each} agent $i \in \vec{t}$}
        \If{$\ver(i) \neq 1$} $L \leftarrow L \union \{ i \}$ \EndIf
    \EndFor
    \If{$L \neq \emptyset$} \Return $\Expo^\alpha(\vec{x}_{-L}, \vec{s}_{-L})$
    \Else\ \Return outcome $j$
    \EndIf
\end{algorithmic}\end{algorithm}

\subsection{The Approximation Guarantee of Exponential}
\label{s:app:exponential-approximation}

As for the approximation guarantee of Exponential, we let $j$ be the outcome of maximum total weight and let $\OPT = w_j = \norm{\vec{w}}_\infty$ be $j$'s weight (and the optimal social welfare). We recall that the Exponential allocation rule with parameter $\alpha > 0$ is the solution to the following optimization problem
\[
 \vec{p}^\ast = \arg\max_{\vec{p} \in \Delta(O)}
 \sum_{i \in N} \vec{x}_i \cdot \vec{p} + \alpha H(\vec{p})\,, \]
where $H(\vec{p}) = - \sum_j p_j \ln p_j$ is the informational entropy of $\vec{p}$ (see e.g., \cite{HK12}). So, we let $\vec{p}^\ast$ be the probability distribution of $\Expo^\alpha$ and let $\vec{e}_j$ be the unit vector with a single $1$ in coordinate $j$. Then,
%
%\begin{align*}
% \OPT = \sum_{i \in N} \vec{x}_i \cdot \vec{e}_j + \alpha H(\vec{e}_j) & \leq
%
% \sum_{i\in N} \vec{x}_i \cdot \vec{p}^\ast + \alpha H(\vec{p}^\ast) \\
%
% & \leq \sum_{i\in N} \vec{x}_i \cdot \vec{p}^\ast + \alpha \ln m
%\end{align*}
\[
 \OPT = \sum_{i \in N} \vec{x}_i \cdot \vec{e}_j + \alpha H(\vec{e}_j) \leq
 \sum_{i\in N} \vec{x}_i \cdot \vec{p}^\ast + \alpha H(\vec{p}^\ast)
 \leq \sum_{i\in N} \vec{x}_i \cdot \vec{p}^\ast + \alpha \ln m
\]
The first inequality follows from the definition of $\vec{p}^\ast$ and the second inequality holds because the maximum entropy of any distribution in $\Delta(O)$ is $\ln m$.
Therefore, the expected welfare of $\Expo^\alpha$ is $\sum_{i\in N} \vec{x}_i \cdot \vec{p}^\ast \geq \OPT - \alpha \ln m$. \qed